\documentclass[10pt]{article}
\usepackage[margin=0.7in,footskip=0.25in]{geometry}
\date{}
\usepackage[shortlabels]{enumitem}

\usepackage{amsfonts}
\usepackage{amsmath, amsthm, amssymb, dsfont,stmaryrd}
\usepackage{graphicx}
\usepackage{comment, cite}
\usepackage{psfrag}
\usepackage{bbm}
\usepackage{csquotes}

\newcommand{\mkv}{-\!\!\!\!\minuso\!\!\!\!-}

\usepackage[dvipsnames]{xcolor}
\usepackage{tikz}

\usetikzlibrary{fit}					
\usetikzlibrary{backgrounds}

\newtheorem{theorem}{Theorem}
\newtheorem{lemma}{Lemma}

\newtheorem{proposition}{Proposition}
\newtheorem{conjecture}{Conjecture}
\newtheorem{corollary}{Corollary}
\newtheorem{definition}{Definition}
\newtheorem{example}{Example}

\newtheorem{remark}{Remark}

\definecolor{darkblack}{rgb}{0, .07, .5}
\definecolor{darkred}{rgb}{0.5,0,0}
\definecolor{mahogany}{rgb}{0.65, 0., 0.5}

\newcommand{\RN}[1]{%
	\textup{\uppercase\expandafter{\romannumeral#1}}%
}

\title{On the Source Model Key Agreement Problem}
\author{Hamidreza Abin and Amin Gohari\\
Department of Information Engineering\\
The Chinese University of Hong Kong\\
Sha Tin, NT, Hong Kong\\
\{hamabin,agohari\}@ie.cuhk.edu.hk
}

\allowdisplaybreaks
\begin{document}
	\maketitle
  
\begin{abstract}
We consider the source model key agreement problem involving two legitimate parties and an eavesdropper who observes 
$n$ i.i.d.\ samples of $X$ and $Y$ and $Z$, respectively. In this paper, we focus on one of the simplest instances where the key capacity remains open, specifically when \( X \) and \( Y \) are binary random variables and \( Z \) is a function of the pair \( (X, Y) \). The best-known upper bound on the key capacity is characterized by an inf-max optimization problem that generally lacks a closed-form solution. We provide general conditions under which the upper bound reduces to $I(X; Y)$. As an example, we consider the XOR setting in which 
$X$ and $Y$ are binary, and $Z$
 is the XOR of $X$ and $Y$. The upper bound reduces to $I(X;Y)$ for this source. Next, we conjecture that the rate $I(X; Y)$ is not achievable for the XOR source, and provide some ideas that might be useful for developing a new upper bound on the source model problem.
\end{abstract}

\section{Introduction}

The two-party source model problem consists of two legitimate parties, Alice and Bob, who aim to establish a shared secret key hidden from an eavesdropper, Eve \cite{1,10}. Alice, Bob, and Eve observe 
$n$ i.i.d.\ samples of $X$ and $Y$ and $Z$\textcolor{black}{,} respectively. Random variables $X,Y$\textcolor{black}{,} and $Z$ are distributed according to $P_{X,Y,Z}$, known to all parties. Through authenticated but public communication, Alice and Bob collaborate to generate a key that remains secret from Eve (extensions to multiple parties and continuous random variables can be found in \cite{CsNar,Chanomni,Nitinawarat,Tyagi,Chan2,chan2018optimality}; please see \cite{elgamal} for a detailed treatment). 
 In a non-interactive setting, Alice and Bob locally generate messages $F_1$ and $F_2$ from $P_{F_1|X^n}$ and $P_{F_2|Y^n}$, respectively, and exchange these messages.
 Then Alice and Bob create their own key from 
$K_A\sim P_{K_A|X^n,F_1,F_2}$ and  $K_B\sim  P_{K_B|Y^n,F_1,F_2}$, respectively.  On the other hand, in the interactive setting,  Alice initiates the process by generating a public message $F_1$ using the distribution $P_{F_1|X^n}$
 and shares it with Bob. Bob then responds with $F_2$
 derived from $P_{F_2|Y^n,F_1}$. This exchange continues for $r$ rounds, during which both senders alternately generate and share public messages
 $F_3,F_4,\cdots, F_r$. At the end of these interactions, Alice derives a key $K_A$	based on
 $P_{K_A|X^n,F_{1:r}}$, while Bob produces a corresponding key $K_B$ using $P_{K_B|Y^n,F_{1:r}}$.
In an
$(n,R,\delta)$ code, the keys must match with probability $1-\delta$:
\begin{align}
    \mathbb{P}[K_A=K_B]\geq 1-\delta,\label{pd1}
\end{align}
where 
$\delta$ is a small positive value.
The key rate is defined as $R=\frac{1}{n}H(K_A)$. $R$ must exhibit high entropy, quantified as:
\begin{align}
    R\geq \frac{1}{n}\log(|\mathcal{K}|)-\delta,\label{pd2}
\end{align}
where  
$\mathcal{K}$ denotes the alphabet of 
$K_A$ and $K_B$.  Furthermore, the keys should be nearly independent of Eve’s observations, formalized by:
\begin{align}
    \frac{1}{n}I(K_A;Z^n,F_{1:r})\leq \delta.\label{pd3}
\end{align}
Note that the value of $r$ for the non-interactive regime is $2$. We note that if we omit the factor \( 1/n \) in equation \eqref{pd3}, we obtain a seemingly stronger notion of secrecy; however, this is equivalent to the weaker form presented in \eqref{pd3}. For a proof of equivalence, see \cite{maurer2000information}.

We say that the rate $R$ is achievable if for every $\delta>0$, there is an $(n,\delta,R)$ code. 
The maximum achievable key rate is known as the source model secret-key (SK) capacity. It is denoted by 
$S(X;Y\|Z)$ in the interactive setting, and by
$S_{\text{ni}}(X;Y\|Z)$ in the non-interactive setting.

{\color{black}
Originally studied by Ahlswede and Csiszár~\cite{1} and Maurer~\cite{10}, determining the secret key capacity remains a fundamental open problem in network information theory. For an introduction to this problem, the reader may refer to Chapter~22 of \cite{elgamal}.}
In general scenarios, the key capacity remains unknown (both in the interactive and non-interactive communication settings) with the exception of certain special cases, such as when \(X, Y,\) and \(Z\) form a Markov chain in some order. As a side result, we derive the capacity for a new class of sources, which we call ``deterministic-erasure'' sources in Appendix \ref{appndx}.

One scenario where the capacity is known occurs when communication is restricted to one-way transmission from one party to another.
 The one-way SK capacity from $X$ to $Y$ equals \cite{1}:
\begin{align}
&S_{\text{ow}}(X\rightarrow Y\|Z)=\max_{(U,V)\mkv X\mkv (Y,Z)} I(U;Y|V) - I(U;Z|V),
\end{align}
where the conditions $|\mathcal{V}|\leq |\mathcal{X}|$ and $|\mathcal{U}|\leq |\mathcal{X}|$ can be imposed on $\mathcal{V}$, $\mathcal{U}$, and $\mathcal{X}$ representing the alphabets associated with $V$, $U$, and $X$, respectively \cite{elgamal}.\footnote{We can equivalently express \begin{align}
&S_{\text{ow}}(X\rightarrow Y\|Z)=\max_{V\mkv U\mkv X\mkv (Y,Z)} I(U;Y|V) - I(U;Z|V),
\end{align}This follows from noting that changing $U$ to $(U,V)$ will not change the mutual information terms, while introducing the additional Markov chain constraint. However, this comes at the cost of a larger cardinality bound on $U$.}
It is evident that $S_{\text{ow}}(X\rightarrow Y\|Z)$ establishes a lower bound on the SK capacity $S(X;Y\|Z)$. {\color{black}In particular, setting $U=X$ and $V$ constant yields
$$S(X;Y\|Z)\geq I(X;Y)-I(X;Z).$$
The best known lower bound is obtained in \cite{gohari-terminal1} as follows: consider random variables $U_1, U_2, \ldots, U_k$ satisfying the Markov chain conditions
\begin{align}
    & U_i \mkv (X, U_{1:i-1}) \mkv (Y,Z) \quad \text{for odd } i, \label{eq:LowerboundoddiMarkov}\\
    & U_i \mkv (Y,U_{1:i-1}) \mkv (X,Z) \quad \text{for even } i. \label{eq:LowerboundeveniMarkov}
\end{align}
For any integers $1 \leq m \leq k$, the secret key capacity satisfies $S(X;Y\|Z)$ is bounded from below by
\begin{align}
\Bigg[\sum_{\substack{i \geq m \\ \text{odd } i}} I(U_i; Y | U_{1:i-1}) - I(U_i; Z | U_{1:i-1})\Bigg] 
    + \Bigg[\sum_{\substack{i \geq m \\ \text{even } i}} I(U_i; X | U_{1:i-1}) - I(U_i; Z | U_{1:i-1})\Bigg]. \label{eq:lowerbound}
\end{align}
This bound is difficult to evaluate since both $m$ and $k$ are arbitrary and evaluation of the bound becomes increasingly difficult for large $m,k$. 
}

An early upper bound on the key capacity was \cite{10,1}
$$S(X;Y\|Z)\leq I(X;Y|Z).$$
Authors in \cite[pp. 1126, Remark 2]{1},\cite{MaurerWolf99}
improved this upper bound by realizing that degrading Eve cannot decrease the key rate, i.e., for any degradation $J$ of $Z$ through a Markov chain $J\mkv Z\mkv (X,Y)$ we have
$$S(X;Y\|Z)\leq S(X;Y\|J)\leq I(X;Y|J).$$
The quantity
\begin{align}I(X;Y\downarrow Z)\triangleq \min_{J\mkv Z\mkv (X,Y)}I(X;Y|J)\label{eqnInT}\end{align}
is called  the intrinsic information.  
More generally, we can relax the degraded assumption $J\mkv Z\mkv (X,Y)$ to a less noisy condition. The inequality 
$$S(X;Y\|Z)\leq S(X;Y\|J)$$ holds if the channel $P_{Z|X,Y}$ is less noisy than the channel $P_{J|X,Y}$. Even more generally, for any arbitrary $P_{J|X,Y}$, we  have {\color{black}\cite[Corollary 2]{gohari-terminal1}}
\begin{align}
S(X;Y\|Z)-S(X;Y\|J)\leq \max_{(U,V)\mkv (X,Y)\mkv (J,Z)} I(U;J|V)-I(U;Z|V)=S_{\text{ow}}(X,Y \rightarrow Z|J)\label{eqnd1}
\end{align} 
Observe that the right\textcolor{black}{-}hand side of the above equation vanishes when $P_{Z|X,Y}$ is less noisy than the channel $P_{J|X,Y}$, and was interpreted as ``the penalty of deviating
from the less-noisy condition'' in \cite{gohari2020coding}. Using \eqref{eqnd1} and the upper bound
$$
S(X;Y\|J)\leq I(X;Y|J)
$$
the following upper bound was proposed in {\color{black}\cite[Corollary 2]{gohari-terminal1}}:
 \begin{align}
S(X;Y\|Z)\leq \Psi(X;Y\|Z)&\triangleq\inf_J\max_{(U,V)\mkv (X,Y)\mkv (J,Z)} I(X;Y|J)+I(U;J|V)-I(U;Z|V)
    \label{main1}
 \end{align}

 Note that $J$ serves as an auxiliary receiver in the sense of \cite{gon21}. 
 The above bound is the best known upper \textcolor{black}{bound} on the key capacity. It involves a min-max optimization problem, which makes computing the upper bound challenging. In fact, it is not known whether the infimum over $J$ in \eqref{main1} is a minimum (see \cite{gohari-terminal1} for a discussion on the computability of this bound). A weaker, but more convenient upper bound was also proposed in  \cite{gohari-terminal1}

\begin{proposition}[\cite{gohari-terminal1}]\label{lem1}
    For any $P_{X,Y,Z}$ we have:
    \begin{align}
   S(X;Y\|Z)&\leq \inf_J\max_{(U,V)\mkv (X,Y)\mkv (J,Z)} I(X;Y|J)+I(U;J|V)-I(U;Z|V)\label{main2-1}
 \\&\leq \min_J I(X;Y|J)+I(X,Y;J|Z).\label{main2}
\end{align}
\end{proposition}
We denote 
\begin{align}\hat{\Psi}(X;Y\|Z)\triangleq\min_{J}I(X;Y|J)+I(X,Y;J|Z).\label{eqnAAAApsihat}\end{align} 
Notice that the upper bound $\Psi(X;Y\|Z)$ in \eqref{main1} depends solely on 
$P_{X,Y,J}$ and $P_{X,Y,Z}$, meaning the minimization is restricted to  $P_{J|X,Y}$. However, the weaker  upper bound $\hat{\Psi}(X;Y\|Z)$ in \eqref{main2}  requires minimization over 
 $P_{J|X,Y,Z}$. On the other hand, the weaker upper bound $\hat{\Psi}(X;Y\|Z)$ does not involve the auxiliary random variables $U$ and $V$.

{\color{black}
\begin{remark}
There are notable parallels between key agreement in classical and quantum scenarios \cite{gisin2002linking,gisin2000linking}. Furthermore, the quantum counterparts of the upper bounds in the key agreement problem represent independently interesting objects of study; while intrinsic information serves primarily as an upper bound on $S(X;Y\|Z)$ in the classical setting, its quantum counterpart—squashed entanglement—acts as a meaningful measure of entanglement for bipartite quantum systems \cite[Section III]{christandl2004squashed}. On the one hand, the intrinsic information upper bound has led to the definition of squashed entanglement. On the other hand, insights from quantum information theory have been used to derive new upper bounds for the classical key capacity \( S(X;Y\|Z) \), as demonstrated in \cite{gisin2001bound, renner2003new}. In \cite{keykhosravi2014source}, it is demonstrated that in the classical framework, if we permit \( J \) to be a quantum system while maintaining \( (X, Y, Z) \) as classical, the upper bound on \( S(X; Y \| Z) \) in equation \eqref{eqnAAAApsihat} remains valid. However, it is important to note that the minimum in  \eqref{eqnAAAApsihat} should be replaced by an infimum, as there is currently no known dimension bound on the size of the quantum system \( J \) (see \cite{beigi2013dimension}).\label{rmks45}
\end{remark}

\begin{remark} 
The question of when key agreement is possible, i.e., when $S(X;Y \| Z) > 0$, has been an important topic in the literature \cite{10,MaurerWolf99,OrlitskyWigderso,gohari2020coding}.
Orlitsky and Wigderson show that if $S(X;Y \| Z) > 0$, just two rounds of interactive communication suffice to achieve a positive key rate~\cite{orlitsky1993secrecy}. 
Although multi-letter characterizations exist for the condition \( S(X;Y \| Z) = 0 \)~\cite{gohari2020coding,orlitsky1993secrecy}, 
these results are not in single-letter form and thus cannot be used to verify whether \( S(X;Y \| Z) = 0 \). 
The only known method to prove $S(X;Y \| Z) = 0$ is to show that one of the known upper bounds on $S(X;Y\|Z)$ vanishes. 
It is indeed conjectured in~\cite{MaurerWolf99} that $S(X;Y \| Z) = 0$ if and only if the intrinsic information vanishes, i.e., $I(X;Y \downarrow Z) = 0$. Given the upper bound  in Proposition~\ref{lem1}, the condition \( I(X;Y \downarrow Z) = 0 \) in the conjecture should be refined to state that there exists a random variable \( J \) such that (i) \( X \rightarrow J \rightarrow Y \) forms a Markov chain, and (ii)
    the channel \( P_{Z|X,Y} \) is less noisy than the channel \( P_{J|X,Y} \).
\end{remark}

{\color{black}
\subsection{Overview of our contributions}

Consider the upper bound
\[
S(X;Y\|Z)\leq \Psi(X;Y\|Z)
= \inf_{J}\; \max_{(U,V)\mkv (X,Y)\mkv (J,Z)} \Big\{ I(X;Y| J) + I(U;J| V) - I(U;Z| V) \Big\}.
\]
No cardinality bound on the alphabet of $J$ is known. However, even after fixing a choice of $P_{J| X,Y,Z}$, the cardinality of the alphabets of $U$ and $V$ in the inner maximization can be bounded from above by $|\mathcal{X}||\mathcal{Y}|$. Thus, after fixing $J$, one still needs to solve an optimization over $P_{U,V| X,Y,Z}$ with  $|\mathcal{X}||\mathcal{Y}||\mathcal{Z}|\cdot (|\mathcal{X}||\mathcal{Y}|)^2$ free parameters, rendering a direct evaluation of the upper bound practically infeasible.

Our first observation is that when $Z$ is a function of $(X,Y)$, the above upper bound coincides with the simpler upper bound
\[
S(X;Y\|Z)\leq \hat{\Psi}(X;Y\|Z) = \min_{J}\; I(X;Y| J) + I(X,Y;J| Z),
\]
which eliminates the auxiliary variables $(U,V)$ in the inner maximization. When $|\mathcal{X}|$ and $|\mathcal{Y}|$ are small and $Z$ is a function of $(X,Y)$, it is feasible to numerically evaluate $\hat{\Psi}(X;Y\|Z)$; in this case one can impose a cardinality bound $|\mathcal{J}|\leq |\mathcal{X}||\mathcal{Y}|$ on $J$. However, $I(X;Y|J)+I(X,Y;J|Z)$ is neither convex nor concave in $P_{J| X,Y,Z}$. Consequently, even if one has a good guess for the minimizing $P_{J| X,Y,Z}$, it may be hard to certify it as a global minimizer.

Our first contribution is a ``weak duality'' technique to address this certification problem. We believe the technique is of independent interest and may be useful for other non-convex optimization problems in network information theory.

\subsubsection{Main Contribution 1}

Many problems in network information theory involve rate regions or non-convex optimizations over auxiliary random variables. As an example, consider the broadcast channel: Marton's inner bound involves three auxiliary random variables representing private and common messages. For certain structured channels (e.g., semi-deterministic or erasure Blackwell channels), the global optimizers take a particular form. In some cases, this structure is provable; in others, it is only supported by numerical evidence. To address questions of this type, we propose a ``weak-duality'' idea and illustrate it in the source model.

\textbf{Statement of the issue:}
Let \(\mathcal{P}\) denote the set of joint distributions \(P_{X,Y,Z}\) for which the global minimizer of \(\hat{\Psi}(X;Y\|Z)\) is a constant \(J\). Our goal is to characterize the set \(\mathcal{P}\).  Equivalently, we seek to identify all \(P_{X,Y,Z}\) satisfying  
\[
I(X;Y) - I(X;Y|J) - I(X,Y;J|Z) \leq 0, \quad \forall\, P_{J|X,Y,Z}.
\]
The equality in this inequality holds when \(J\) is a constant. Hence, we may rewrite the condition as  
\[
I(X;Y) - H(X,Y|Z) = \min_{J} \big[ I(X;Y|J) + H(X,Y|J,Z) \big].
\]

Observe that \(\min_{J} \big[ I(X;Y|J) - H(X,Y|J,Z) \big]\) computes the \emph{lower convex envelope} of the mapping  
\begin{align}
Q_{X,Y,Z} \mapsto I_Q(X;Y) - H_Q(X,Y|Z) \label{mapping2}
\end{align}  
evaluated at \(Q_{X,Y,Z} = P_{X,Y,Z}\) (see \cite{nair2013upper} for details on the relation between lower convex envelopes and auxiliary random variables). The mapping in \eqref{mapping2} is defined for all \(Q_{X,Y,Z}\) on fixed alphabets \(\mathcal{X} \times \mathcal{Y} \times \mathcal{Z}\).  

Given a function \(\varphi(\cdot)\), its lower convex envelope \(\mathcal{K}[\varphi]\) is defined as the largest convex function dominated by \(\varphi\), i.e.,  
\[
\mathcal{K}[\varphi](x) = \sup \{ \tilde{\varphi}(x) \mid \tilde{\varphi} \text{ convex and } \tilde{\varphi} \leq \varphi \}.
\]  
See Figure \ref{fig:lcews} for an illustration. 

The set \(\mathcal{P}\) consists precisely of those distributions \(P_{X,Y,Z}\) where the mapping in \eqref{mapping2} coincides with its lower convex envelope. For a general non-convex function \(\varphi(x)\), computing \(\mathcal{K}[\varphi](x^\star)\) at a point \(x^\star\) requires global knowledge of \(\varphi(\cdot)\); local information about its value near \(x^\star\) is insufficient. Thus, determining the set \(\{x : \mathcal{K}[\varphi](x) = \varphi(x)\}\) necessitates a global access to the function \(\varphi(\cdot)\). However, given that the mapping in \eqref{mapping2} is not arbitrary but constructed from Shannon entropy measures, can this structural property enable us to derive global conclusions from local analysis?  

\begin{figure}  
    \centering
    \includegraphics[width=0.7\textwidth]{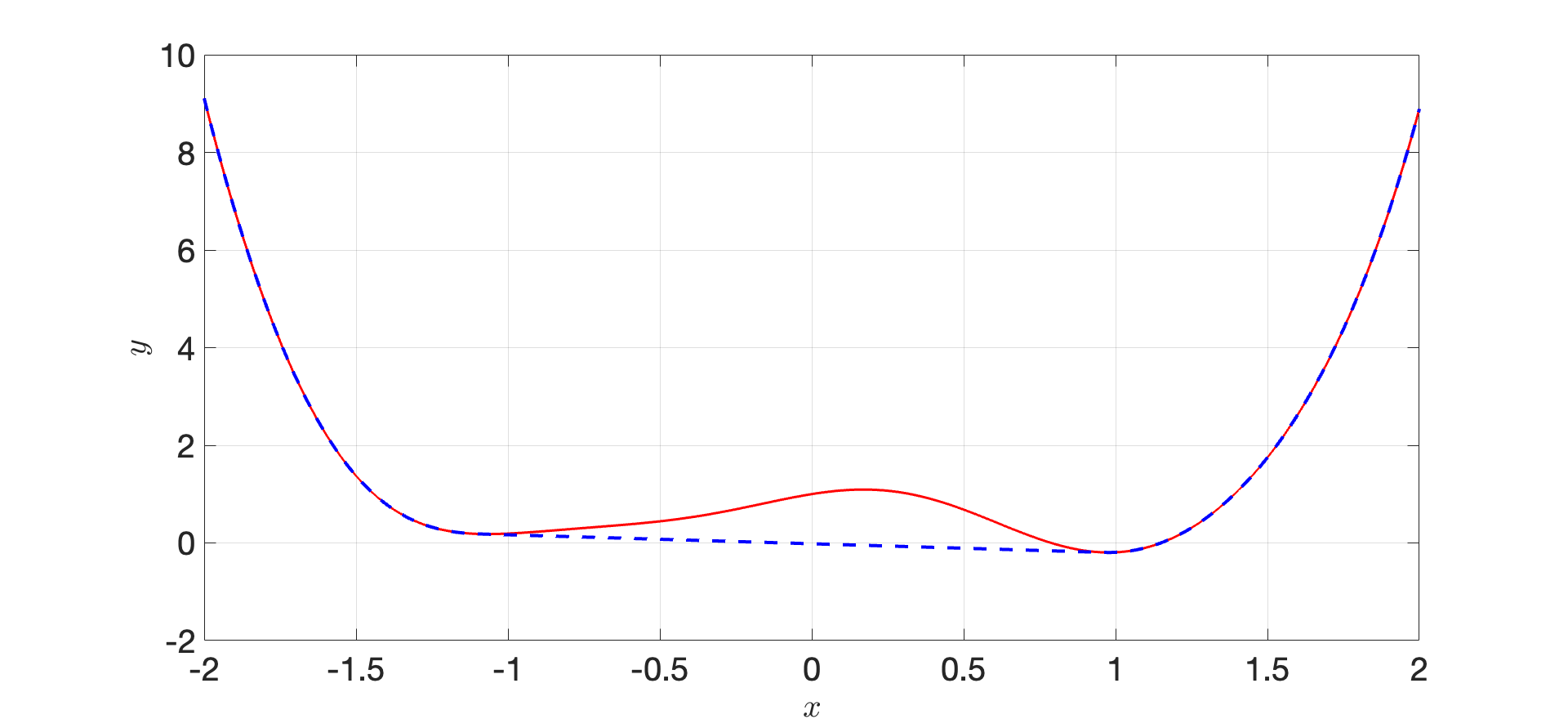}  
  \caption{Visualization of a function and its lower convex envelope. The function coincides with its lower convex envelope for $x\geq 1$ and $x\leq -1$. }
\label{fig:lcews}  
\end{figure}

\textbf{From local to global:}
Our approach to bridge global to local is as follows. Define
\[
f(X,Y,Z,J) \triangleq I(X;Y) - I(X;Y| J) - I(X,Y;J| Z).
\]
We seek the set of $P_{X,Y,Z}$ such that
\[
\max_{J} f(X,Y,Z,J)\;\leq\; 0,
\]
where the maximization is over $P_{J| X,Y,Z}$. Suppose we can find another functional $g(X,Y,Z,T)$ such that \underline{for every $P_{X,Y,Z}$} we have
\[
\max_{J} f(X,Y,Z,J)\;\leq\; \min_{T} g(X,Y,Z,T),
\]
where the minimization is over $P_{T| X,Y,Z}$. This is akin in spirit to linear-programming duality, but carried out at the level of mutual-information expressions. If such a $g$ exists, then for a given $P^{\star}_{X,Y,Z}$ one can certify
\[
\max_{J} f(X,Y,Z,J)\;\leq\; 0
\]
by presenting a single witness $P_{T| X,Y,Z}$ for which
\[
g(X,Y,Z,T)=0
\]
holds at $P^{\star}_{X,Y,Z}$. Crucially, only a single witness and a local verification at $P^{\star}_{X,Y,Z}$ is needed.

In this paper, we take
\[
g(X,Y,Z,T) \triangleq I(T;Z) + I(X;Y| T) + I(T;Z| X,Y).
\]
One then needs to verify that for every $P_{X,Y,Z,T,J}$,
\[
I(X;Y) - I(X;Y| J) - I(X,Y;J| Z)
\;\leq\;
I(T;Z) + I(X;Y| T) + I(T;Z| X,Y).
\]
The above is a non-Shannon-type inequality. It can, however, be established using Shannon-type inequalities if we impose the constraint $I(T;J| X,Y,Z)=0$. This constraint can be added without loss of generality because the expression depends only on the marginals $P_{X,Y,Z,J}$ and $P_{X,Y,Z,T}$. See the proof of Theorem~\ref{Th:gen:classic}.

To demonstrate the power of this duality idea, we consider the case where $X$ and $Y$ are binary and $Z$ is a function of $(X,Y)$. We show that a tight characterization of $\mathcal{P}$ can be obtained in the binary case. Namely, if $P_{X,Y,Z}\in\mathcal{P}$, i.e., constant $J$ is optimal, we can explicitly identify a witness $P_{T| X,Y,Z}$ for which $g(X,Y,Z,T)=0$. In addition, we also study cases where $P_{X,Y,Z}\notin\mathcal{P}$, i.e., a non-constant $J$ is the minimizer; here, we find the minimizer $J$ as well and show that the minimizer makes $I(X;Y|J)+I(X,Y;J|Z)=0$. Thus, we have $S(X;Y\|Z)=0$ for $P_{X,Y,Z}\notin\mathcal{P}$.
See Theorem~\ref{th:binary} for details. A particular example of interest is when $Z=X\oplus Y$. The set $\mathcal{P}$ restricted to $Z=X\oplus Y$ is the entire set of \emph{all} joint distributions on $\mathcal{X}\times\mathcal{Y}$ and we obtain $\hat{\Psi}(X;Y\|Z)=\Psi(X;Y\|Z)=I(X;Y)$ for every $P_{X,Y}$, when $Z=X\oplus Y$.

\subsubsection{Main Contribution 2}

Having established that the best known upper bound on $S(X;Y\|Z)$ reduces to $I(X;Y)$ in some cases, our second contribution is evidence of the suboptimality of this upper bound. We also provide a new upper bound for the non-interactive case, which is always no larger than the best known bound. Unfortunately, evaluating our new bound is challenging because it has essentially the same complexity as $\Psi(X;Y\|Z)$, which, as discussed earlier, is computationally intractable.

As discussed earlier, we use the ``weak-duality'' idea to show that when $X$ and $Y$ are binary and $Z=X\oplus Y$, we have $\hat{\Psi}(X;Y\|Z)=\Psi(X;Y\|Z)=I(X;Y)$. Can it be that $I(X;Y)$ is also achievable, i.e., $S(X;Y\|Z)=I(X;Y)$?

\textbf{Achievability of the key rate $I(X;Y)$.}
If $Z$ is a constant random variable, then $S(X;Y\|Z)=I(X;Y)$. One strategy to achieve this rate has one party (say Alice) randomly \textcolor{black}{bin} $X^n$ at the Slepian--Wolf rate of $nH(X| Y)$ bits and communicate the bin index to Bob over the public channel. Bob then recovers $X^n$ and the parties generate $nH(X)$ bits of common randomness. Since $nH(X| Y)$ bits are revealed publicly, they perform privacy amplification and output $n(H(X)-H(X| Y))=nI(X;Y)$ secret bits that remain hidden from the eavesdropper. This protocol is communication-heavy: to generate $nI(X;Y)$ secret bits, it reveals $nH(X| Y)$ bits about $X^n$ on the public channel. A symmetric protocol with Bob speaking reveals $nH(Y| X)$ bits. One may ask whether interactive communication can reduce the public communication needed to achieve $I(X;Y)$. While this is possible in certain instances, the answer is negative in general: the minimum public communication rate to achieve $I(X;Y)$ is characterized in~\cite{Tyagi}. It has been conjectured in~\cite{Tyagi} (and proved in some cases) that when $X$ and $Y$ are binary, at least $n\min\{H(X| Y),H(Y| X)\}$ bits must be revealed publicly to achieve a key rate of $I(X;Y)$.

Now suppose $Z$ is a general, non-constant random variable (i.e., the eavesdropper has side information $Z$ about $(X,Y)$ in addition to access to the public transcript). Clearly, $S(X;Y\|Z)\le I(X;Y)$. But can we still achieve $I(X;Y)$? This specific question does not appear to have been explicitly addressed in the literature. We conjecture that when $X$ and $Y$ are binary and $Z=X\oplus Y$, the rate $I(X;Y)$ is \emph{not} achievable for a general $P_{X,Y}$. To begin, consider the one-way protocols above when $Z=X\oplus Y$. If $H(Z)>\max\{H(X),H(Y)\}$, then $H(X| Z)<H(X| Y)$ and $H(Y| Z)<H(Y| X)$. In this case, the one-way protocols fail because, assuming that the bin index of $X^n$ at rate $H(X|Y)$ is revealed, the eavesdropper can utilize $Z^n$ and a Slepian--Wolf decoder to reconstruct the entire sequence $X^n$.

Next, consider interactive protocols. One strategy is to design an interactive transcript $F_1,F_2,\ldots,F_r$ that yields secret keys such that
\begin{align*}
\mathbb{P}[K_A=K_B] &\approx 1,\\
\frac{1}{n}H(K_A) &\approx I(X;Y),\\
\frac{1}{n}I(K_A;F_{1:r}) &\approx 0,\\
\frac{1}{n}H(Z^n| F_{1:r}) &\approx 0,
\end{align*}
so that the parties both (i) generate keys of maximum rate $I(X;Y)$ that are independent of the transcript $F_{1:r}$, and (ii) interactively compute the XOR $Z^n$ (a function-computation problem). If both goals are met, then $K_A,K_B$ are independent of $(F_{1:r},Z^n)$. But can these two goals be met simultaneously?

First, consider the non-interactive case with two messages $F_1$ and $F_2$ generated via $P_{F_1| X^n}$ and $P_{F_2| Y^n}$, respectively. The goal is to compute $Z^n=X^n\oplus Y^n$ from $(F_1,F_2)$ alone. This is the classical modulo-two sum problem in source coding, originally studied by Körner and Marton~\cite{1056022}. One can import ideas from this line of work (e.g., \cite{sga15,nair2020}) to study the source model problem; this is one of our key ideas leading to a new upper bound in the non-interactive case.

For the interactive case with multiple rounds, we assume
\[
\frac{1}{n}H(Z^n| F_{1:r}) \approx \Delta
\]
for some $\Delta\ge 0$. The case $\Delta=0$ corresponds to an interactive modulo-two sum problem. We provide evidence that $\Delta$ cannot be close to zero when achieving the rate $I(X;Y)$. Moreover, we derive a new upper bound for interactive protocols as a function of a given $\Delta$. To explain this bound, we first overview the existing best-known upper bound.

\textbf{Overview of the proof of the best-known upper bound.}
Roughly speaking, the main steps to prove $S(X;Y\|Z)\le \Psi(X;Y\|Z)$ are as follows:\footnote{The proof sketched here is not the original proof given in \cite{gohari-terminal1}; instead, it follows the add-and-subtract ``auxiliary receiver'' approach; see, e.g., Section~II.A in \cite{wen2024new}.} take an arbitrary $(n,R,\delta)$ code generating keys $K_A\approx K_B$ via interactive communication $F_{1:r}$. Then
\begin{align*}
R \;\approx\; \frac{1}{n}H(K_A) \;\approx\; \frac{1}{n}H(K_A| Z^n,F_{1:r})
&= \frac{1}{n}H(K_A| J^n,F_{1:r})
  + \frac{1}{n}\big(H(K_A| Z^n,F_{1:r}) - H(K_A| J^n,F_{1:r})\big).
\end{align*}
The difference term can be single-letterized:
\begin{align*}
\frac{1}{n}\big(H(K_A| Z^n,F_{1:r}) - H(K_A| J^n,F_{1:r})\big)
&= \frac{1}{n}\big(I(K_A;J^n| F_{1:r}) - I(K_A;Z^n| F_{1:r})\big)\\
&= \frac{1}{n}\sum_{i=1}^{n}\big(I(K_A;J^i,Z_{i+1}^n| F_{1:r}) - I(K_A;J^{i-1}, Z_{i}^n| F_{1:r})\big)\\
&= \frac{1}{n}\sum_{i=1}^{n}\big(I(K_A;J_i| F_{1:r},J^{i-1},Z_{i+1}^n) - I(K_A;J^{i-1} Z_i| F_{1:r},J^{i-1},Z_{i+1}^n)\big)\\
&= I\big(K_A;J_Q| Q,F_{1:r},J^{Q-1},Z_{Q+1}^n\big) - I\big(K_A;Z_Q| Q,F_{1:r},J^{Q-1},Z_{Q+1}^n\big),
\end{align*}
where $Q$ is a time-sharing variable independent of $(K_A,K_B,F_{1:r},X^n,Y^n,Z^n)$. Identifying $U=K_A$ and $V=(Q,F_{1:r},J^{Q-1},Z_{Q+1}^n)$ yields
\[
\frac{1}{n}\big(H(K_A| Z^n,F_{1:r}) - H(K_A| J^n,F_{1:r})\big)
\;\le\;
\max_{(U,V)\mkv (X,Y)\mkv (J,Z)} \big\{ I(U;J| V) - I(U;Z| V) \big\}.
\]
Next,
\begin{align*}
\frac{1}{n}H(K_A| J^n,F_{1:r})
&\approx \frac{1}{n}I(K_A;K_B| J^n,F_{1:r})\\
&\le \frac{1}{n}I(X^n;Y^n| J^n,F_{1:r})\\
&\le \frac{1}{n}I(X^n;Y^n| J^n)\\
&= I(X;Y| J),
\end{align*}
where the inequality $I(X^n;Y^n| J^n,F_{1:r})\le I(X^n;Y^n| J^n)$ can be proved by induction on $r$.

\textbf{Improvement to the upper bound.}
Our new idea is to exploit
\[
\frac{1}{n}H(Z^n| F_{1:r}) \approx \Delta
\]
and impose a corresponding constraint on the auxiliary random variable $V=(Q,F_{1:r},J^{Q-1},Z_{Q+1}^n)$. This is formalized in Theorem~\ref{tight:the1}. We provide explicit upper bounds in the non-interactive case. Obtaining explicit bounds for the interactive case requires explicit bounds on $\Delta$ for protocols achieving the key rate $I(X;Y)$, which we leave for future work.

We also discuss other ideas for improving the upper bound. For instance, we show that using multi-letter versions of $\hat{\Psi}$ does not yield improved bounds. Moreover, as discussed above, the identity
\eqref{eqnd1} is the key step in deriving the best-known upper bound: in that step we replace $Z$ with an auxiliary receiver $J$. To generalize this argument, one may also replace $X$ and $Y$ with other variables. We propose the following conjecture: for any $P_{X,Y,Z,X',Y',Z'}$,
\[
S(X;Y\|Z)-S(X';Y'\|Z')\;\leq\; S_{\text{ow}}(X,Y\!\to\! Z'\|Z)
+ I(Y',Z';X| X') + I(X',Z';Y| Y') + I(X;Y| X',Y',Z').
\]
This conjecture generalizes \eqref{eqnd1} (it reduces to \eqref{eqnd1} by setting $X=X'$ and $Y=Y'$).

\color{black}

\textbf{Summary of our contributions:} 
The main results of this paper can be summarized as follows:
\begin{itemize}
    \item We observe that $\hat{\Psi}(X;Y\|Z)=\Psi(X;Y\|Z)$ if $Z$ is a function of $(X,Y)$. This is immediate by the choice of $U=X,Y$ and $V=Z$ in \eqref{main2-1} which would reduce the expression to that in \eqref{main2}.
    \item We find sufficient conditions {\color{black} provided  in Theorem \ref{Th:gen:classic} and Lemma \ref{lem:con-hull}} that imply that
    $$\hat{\Psi}(X;Y\|Z)=\Psi(X;Y\|Z)=I(X;Y).$$
    In particular, the sufficient conditions  are satisfied when $X$ and $Y$ are binary random variables (with an arbitrary joint distribution) and  $Z$ is the XOR of $X$ and $Y$ {\color{black}(see Theorem \ref{th:binary})}. We call this setting the XOR case. {\color{black} We also discuss other binary operations such as AND or OR in Theorem \ref{th:binary}}.
    \item 
    We conjecture that $S(X;Y\|Z)\neq I(X;Y)$ in the XOR case (see Section \ref{dis:tightness}). 
        Showing this would require developing a new upper bound for the source model problem. {\color{black} The XOR example is particularly interesting because it presents a clean case where previously known techniques -- such as reducing Eve's information via intrinsic mutual information -- fail to improve the basic $I(X;Y)$ upper bound. This demonstrates the need for entirely new ideas. }
        Partial results are also provided in our attempt to develop a new upper bound. {\color{black} By adapting techniques from the modulo-two sum problem in source coding, we derive an explicit bound for the non-interactive case. However, this bound is unfortunately difficult to evaluate.}
\end{itemize}

{\color{black} 
We conclude the introduction with a brief historical remark. An early upper bound on $S(X;Y\|Z)$ is the intrinsic-information bound, defined in~\eqref{eqnInT}. It evaluates the common information between Alice's and Bob's sources under Eve's most favorable choice among all of her possible ``perspectives'' (i.e., post-processings of $Z$). Although this bound was originally conjectured to be tight~\cite{MaurerWolf99}, Renner and Wolf~\cite{renner2003new} constructed a sophisticated counterexample showing that it can be strictly suboptimal. Conceptually, while the intrinsic-information bound weakens Eve by allowing her to degrade her observation, the upper bound in~\cite{renner2003new} moves in the opposite direction by endowing Eve with additional information. Subsequently, by introducing the idea of a fictitious terminal and developing another intricate example,~\cite{gohari-terminal1} established the suboptimality of the bound in~\cite{renner2003new}. In this paper we show that, for the simple binary XOR source, none of these increasingly elaborate techniques for upper bounding $S(X;Y\|Z)$ improves upon the naive bound $I(X;Y)$, which is intuitively unlikely to be tight. In contrast to prior counterexamples of growing complexity, the XOR setting offers a clean and simple example; ideas developed for this source are likely to extend to more general sources as well.
}

The rest of this paper is organized as follows: Section \ref{sec:main:re} evaluates the upper bound $\Psi(X;Y\|Z)$ for some settings. Section \ref{sec:discussionGapp} discusses the gap between $S(X;Y\|Z)$ and $\Psi(X;Y\|Z)$, and provides a conjecture on $S(X;Y\|Z)$. Finally, Section  \ref{dis:tightness} discusses the tightness of the upper bound in the XOR case.
Additionally, in Appendix \ref{appndx}, we compute the secret key capacity for a class of sources where $Z$ is a function of $X$ and $P_{Y|X}$ corresponds to an erasure channel.


\emph{Notation and Definitions:} Random variables are represented by uppercase letters, whereas their realizations are indicated by lowercase letters. It is assumed that all random variables possess finite alphabets. Sets are denoted by calligraphic letters. 
 The notation $X\mkv Y\mkv Z$ signifies that $X$ and $Z$ are conditionally independent given $Y$, which implies that $P_{X,Z|Y}=P_{X|Y}P_{Z|Y}$. In this scenario, we assert that $X\mkv Y\mkv Z$ constitutes a Markov chain. The sequences $(U_1, U_2, \ldots, U_{i})$  and $(U_i,U_{i+1},\cdots,U_n)$ are denoted as $U^i$ and $U_{i}^{n}$ respectively. For a sequence $X^n=(X_1, X_2,\cdots, X_n)$, we use $X_{\backslash i}$ to denote $(X^{i-1},X_{i+1}^{n})$. The entropy of a random variable \(X\) is represented as \(H(X)\), while the mutual information between \(X\) and \(Y\) is denoted by \(I(X;Y)\). We say $Q_{X}\ll P_{X}$ if $Q_X(x)=0$ whenever $P_X(x)=0$. 


\section{Optimality of constant $J$}\label{sec:main:re}
The upper bound \begin{align}
\hat{\Psi}(X;Y\|Z)&=\min_{J}I(X;Y|J)+I(X,Y;J|Z).
 \end{align}
 satisfies \begin{align}\hat{\Psi}(X;Y\|Z)\leq I(X;Y).\label{eqntu}\end{align} This follows from setting $J$ constant. The main results of this section identify sufficient conditions for equality in \eqref{eqntu}. 

\begin{theorem}\label{Th:gen:classic}
Take some arbitrary $P_{X,Y,Z}$. Assume that one can find an auxiliary random variable $T$ on some arbitrary alphabet such that 
$I(T;Z|X,Y)=I(X;Y|T)=I(T;Z)=0$. Then 
$$\hat{\Psi}(X;Y\|Z)=I(X;Y).$$
\end{theorem}
\begin{proof}
Take an arbitrary $P_{J|X,Y,Z}$. Consider the coupling 
    $$P_{T,X,Y,Z,J}\triangleq P_{T|X,Y,Z}P_{J|X,Y,Z}.$$
    Thus, we have
        $$I(T;Z)=I(X;Y|T)=I(T;J|X,Y)=I(T;J|X,Y,Z)=0.$$
        The proof follows from the following identity:
\begin{align*}
    &I(X;Y)-\left\{I(X;Y|J)+I(X,Y;J|Z)\right\}\\&={I(T;Z)+I(X;Y|T)+I(T;Z|X,Y)+2I(T;J|X,Y,Z)}
    \\&\qquad-\left\{I(X;Y|J,T)+I(X,Y;J|T,Z)\right\}\\&\qquad-I(J;T|X)-I(J;T|Y)-I(T;Z|J)-I(T;Z|X,Y,J)
    \\&\leq 0.
\end{align*}
\end{proof}
\begin{remark}\label{remarl:1}
$\hat{\Psi}(X;Y\|Z) = I(X;Y)$ is equivalent to the inequality 
\[
I(X;Y|J) + I(X,Y;J|Z) \geq I(X;Y), \quad \forall P_{J|X,Y,Z}.
\]
When $Z$ is a function of $(X,Y)$, this inequality can be equivalently written as
\[
I(J;X,Y,Z) \geq \frac{1}{2} I(J;X) + \frac{1}{2} I(J;Y) + \frac{1}{2} I(J;Z), \quad \forall P_{J|X,Y,Z}.
\]
This condition holds if and only if $(\tfrac{1}{2}, \tfrac{1}{2}, \tfrac{1}{2})$ belongs to the hypercontractivity ribbon for random variables $(X,Y,Z)$. Recall that the hypercontractivity ribbon for random variables $X_1, X_2, \ldots, X_k$ is defined as the set of $\lambda_i$ such that for any $P_{U|X_{[k]}}$ we have
\[
I(X_{[k]}; U) \geq \sum_i \lambda_i I(X_i; U).
\]
\end{remark}
\begin{remark} 
Theorem \ref{Th:gen:classic} readily extends to quantum systems: consider quantum systems $X,Y,Z,J$ with some arbitrary joint density matrix $\rho_{X,Y,Z,J}$. The inequality
\[
I(X;Y|J) + I(X,Y;J|Z) \geq I(X;Y)
\]
holds if there exists an extended state $\rho_{X,Y,Z,J,T}$ (whose reduced density on $X,Y,Z,J$ is the given state $\rho_{X,Y,Z,J}$) such that
\[
I(T;Z) = I(X;Y|T) = I(T;J|X,Y) = I(T;J|X,Y,Z) = 0.
\]
{\color{black} Thus, Theorem \ref{Th:gen:classic} is applicable to the quantum bound presented in \cite{keykhosravi2014source} as well (see Remark \ref{rmks45}).
}
\end{remark}

The following lemma provides an equivalent condition for the assumption of Theorem \ref{Th:gen:classic} in the special case where \(Z\) is a function of \((X,Y)\).

\begin{lemma}\label{lem:con-hull}
Consider the special case \(Z = f(X,Y)\) for some arbitrary function \(f\). For a given joint distribution \(P_{X,Y}\), there exists a random variable \(T\) such that 
\[
I(Z; T | X,Y) = I(Z; T) = I(X; Y | T) = 0
\]
if and only if \(P_{X,Y}\) lies in the convex hull of the set 
\begin{align}
    \mathcal{A}(P_Z) = \left\{ Q_{X,Y} : I_Q(X;Y) = 0, \quad Q_Z = P_Z \right\}. \label{con:formula}
\end{align}
\end{lemma}

The proof of this lemma can be found in Appendix \ref{append2}.

\subsection{$Z$ is a function of $(X,Y)$}\label{function}

In this subsection, we focus on the scenario where \(Z\) is a deterministic function of \(X\) and \(Y\), i.e., \(Z = f(X,Y)\) for some function \(f\). 

An immediate consequence of this assumption is that the upper bound \(\Psi(X;Y\|Z)\) in \eqref{main2-1} coincides with the simpler bound \(\hat{\Psi}(X;Y\|Z)\) in \eqref{main2}. This follows by choosing the auxiliary random variables \(U = (X,Y)\) and \(V = Z\), which reduces the max-min expression in \eqref{main2-1} directly to the form in \eqref{main2}.

\begin{lemma}
We have $\hat{\Psi}(X;Y\|Z) = {\Psi}(X;Y\|Z)$ if $Z$ is a function of $(X,Y)$.
\end{lemma}

 We will explore the value of $\hat{\Psi}(X;Y\|Z)$ under various scenarios based on the alphabet sizes of $X$ and $Y$.  
\subsubsection{$X$ and $Y$ are binary}
Assume that $X, Y \in \{0,1\}$ are binary random variables. Moreover, 
$$
P_{X,Y} = \begin{bmatrix}
   p_{00} & p_{01} \\
   p_{10} & p_{11}
\end{bmatrix}.
$$

Let $X \oplus Y$ denote the XOR of $X$ and $Y$, and $X \wedge Y$ denote the AND of $X$ and $Y$. 
We consider the three functions $X + Y$, $X \oplus Y$, and $X \wedge Y$ as follows:

$$
Z_1 = X + Y = \begin{cases}
        0 & (X,Y) = (0,0) \\
        1 & (X,Y) \in \{(0,1), (1,0)\} \\
        2 & (X,Y) = (1,1)
    \end{cases}
$$
and
$$
Z_2 = X \oplus Y = \begin{cases}
        1 & (X,Y) \in \{(0,0), (1,1)\} \\
        0 & (X,Y) \in \{(0,1), (1,0)\}
    \end{cases}
$$
and 
$$
Z_3 = X \wedge Y = \begin{cases}
        0 & (X,Y) \in \{(0,1), (1,0), (0,0)\} \\
        1 & (X,Y) = (1,1)
    \end{cases}.
$$
We have the following theorem.
\begin{theorem}\label{th:binary}
The following statements hold for binary random variables $X$ and $Y$:
\begin{itemize}
    \item For any $P_{X,Y}$, we have
    $\Psi(X;Y\|Z_2) = I(X;Y).$
    \item If $\operatorname{Cov}(X,Y) \leq 0$, then
    $\Psi(X;Y\|Z_1) = \Psi(X;Y\|Z_3) = I(X;Y).$  
    \item If $\operatorname{Cov}(X,Y) > 0$, then
    $\Psi(X;Y\|Z_1) = \Psi(X;Y\|Z_3) = 0.$
\end{itemize}
\end{theorem}

The proof of this theorem can be found in Appendix \ref{append2}.

\subsubsection{Ternary-Binary}
In this section, we demonstrate the applicability of Theorem \ref{Th:gen:classic} beyond the binary case by extending the alphabet of $Y$ to $\mathcal{Y} = \{0,1,2\}$.
\color{black}
We give a condition for $\Psi(X;Y\|Z) = I(X;Y)$ when $Z = (X + Y) \mod 2$. 

\begin{theorem}\label{th:ter-bin}
Assume that $\mathcal{X} = \{0,1\}$, $\mathcal{Y} = \{0,1,2\}$, and $\mathcal{Z} = \{0,1\}$, and
\begin{align}\label{Z:function}
   Z = (X + Y) \mod 2 = \begin{cases}
        0 & (X,Y) \in \{(0,0), (1,1), (0,2)\} \\
        1 & (X,Y) \in \{(0,1), (1,0), (1,2)\}
    \end{cases}
\end{align}
If 
$$
\frac{P_{X,Y}(0,0)}{P_Z(0)} + \frac{P_{X,Y}(1,2)}{P_Z(1)} \leq 1
$$
and
$$
\frac{P_{X,Y}(0,2)}{P_Z(0)} + \frac{P_{X,Y}(1,0)}{P_Z(1)} \leq 1,
$$
then
$$
\Psi(X;Y\|Z) = I(X;Y).
$$

\end{theorem}
Proof of this theorem can be found in Appendix \ref{append2}.

\begin{remark}
Note that when $P_Y(2) = 0$ (i.e., $Y$ is binary), $Z$ in \eqref{Z:function} will be the XOR of binary variables $X$ and $Y$. Therefore, Theorem \ref{th:ter-bin} generalizes the first part of Theorem \ref{th:binary}.
\color{black}
\end{remark}

\subsubsection{Arbitrary Functions}
{\color{black}As we observed earlier, when $Z$ is the XOR of $X$ and $Y$, we have $\Psi(X;Y\|Z) = I(X;Y)$ for every source $P_{X,Y}$. Equivalently, the set $\mathcal{P}$ (defined in the introduction) restricted to $Z=X\oplus Y$ is the entire set of \emph{all} joint distributions on $\mathcal{X}\times\mathcal{Y}$. Is there any function, other than XOR, with this property? } This part considers random variables $X, Y,$ and $Z$ with arbitrary alphabet sizes. We aim to determine all functions $f:\mathcal{X} \times \mathcal{Y} \mapsto \mathcal{Z}$ for which the equation $\Psi(X;Y\|Z) = I(X;Y)$ holds for every joint distribution $P_{X,Y}$.
\begin{theorem}\label{thm:thm5}
    Take some fixed alphabets $\mathcal{X}$ and $\mathcal{Y}$. Fix some function $f:\mathcal{X} \times \mathcal{Y} \mapsto \mathcal{Z}$ and assume that $Z = f(X,Y)$. We have
$$
\Psi(X;Y\|Z) = I(X;Y), \qquad \forall P_{X,Y}
$$
if and only if either of the following holds:
\begin{itemize}
    \item $X$ and $Y$ are binary and we have the XOR mapping $f(0,0) = f(1,1) \neq f(1,0) = f(0,1)$,
    \item $f(x,y)$ is a constant function.
\end{itemize}
\end{theorem}
Proof of Theorem \ref{thm:thm5} can be found in Appendix \ref{append2}.
\begin{remark} 
Theorem~\ref{thm:thm5} shows that for every non-constant function other than the XOR function, there exists at least one distribution $P_{X,Y}$
such that $\Psi(X;Y\|Z) \neq I(X;Y)$.
\end{remark}
\color{black}




\section{On the gap between $S(X;Y\|Z)$, $\hat{\Psi}(X;Y\|Z)$ and $I(X;Y)$}\label{sec:discussionGapp}
Let us begin with the upper bound $\hat{\Psi}(X;Y\|Z)$ and $I(X;Y)$. It is clear that
$$
\hat{\Psi}(X;Y\|Z) \leq I(X;Y), \qquad \forall P_{X,Y,Z}.
$$
Assume that 
$$
\hat{\Psi}(X;Y\|Z) = I(X;Y)
$$
for some $P_{X,Y,Z}$. Then, we claim that if $X' = f_1(X)$, $Y' = f_2(Y)$ and $Z' = f_3(Z)$ for some functions $f_1, f_2, f_3$ such that $H(Z'|X',Y')=0$, then we also have
$$
\hat{\Psi}(X';Y'\|Z') = I(X';Y').
$$
We show this by establishing a more general result: 
\begin{theorem}\label{thm6n}
For any $P_{X,Y,Z,X',Y',Z'}$, we have the following inequality:
\begin{align}
\hat{\Psi}(X;Y\|Z) - \hat{\Psi}(X';Y'\|Z') \leq I(X,Y;Z'|Z) + I(Y',Z';X|X') + I(X',Z';Y|Y') + I(X;Y|X',Y',Z'). \label{eqndwe}
\end{align}
\end{theorem}
\begin{remark}
    Assume that $X' = f_1(X)$, $Y' = f_2(Y)$ and $Z' = f_3(Z)$ for some functions $f_1, f_2, f_3$ such that $H(Z'|X',Y')=0$. Then, \eqref{eqndwe} can be written as
    \begin{align}
    I(X;Y) - \hat{\Psi}(X;Y\|Z) \geq I(X';Y') - \hat{\Psi}(X';Y'\|Z'). 
    \end{align}
If we further have $\hat{\Psi}(X;Y\|Z) = I(X;Y)$, we deduce that
\begin{align}
   \hat{\Psi}(X';Y'\|Z') \geq I(X';Y'). 
\end{align}
On the other hand, we always have 
\begin{align}
   \hat{\Psi}(X';Y'\|Z') \leq I(X';Y'). 
\end{align}
Thus, $\hat{\Psi}(X';Y'\|Z') = I(X';Y')$.
\end{remark}
Proof of Theorem \ref{thm6n} is given in Appendix \ref{append2}.

\subsection{\color{black}{Remarks on} Theorem \ref{thm6n}, a Conjecture and Lower bounds on $S(X;Y\|Z)$}\label{sub:future:2}
{\color{black}
It is not known whether there exists some $P_{X,Y,Z}$ such that $S(X;Y\|Z) \neq \hat{\Psi}(X;Y\|Z)$. If $S(X;Y\|Z) = \hat{\Psi}(X;Y\|Z)$ holds for all $P_{X,Y,Z}$, we must have
the following for any $P_{X,Y,Z,X',Y',Z'}$:
\begin{align}
S(X;Y\|Z) - S(X';Y'\|Z') \leq I(X,Y;Z'|Z) + I(Y',Z';X|X') + I(X',Z';Y|Y') + I(X;Y|X',Y',Z'). \label{eqn20nw}
\end{align}
Observe that the above equation is valid when $X' = X, Y' = Y$ since from \eqref{eqnd1} we have
\begin{align}
S(X;Y\|Z) - S(X;Y\|Z') \leq & S_{\text{ow}}(X,Y \rightarrow Z' \| Z) \leq I(X,Y;Z'|Z).
\end{align}
We make the following conjecture:
\begin{conjecture}
    For any $P_{X,Y,Z,X',Y',Z'}$, we have
    $$
    S(X;Y\|Z) - S(X';Y'\|Z') \leq S_{\text{ow}}(X,Y \rightarrow Z' \| Z) + I(Y',Z';X|X') + I(X',Z';Y|Y') + I(X;Y|X',Y',Z').
    $$
    \label{conjecture:1}
\end{conjecture}
This conjecture generalizes \eqref{eqnd1} (it reduces to \eqref{eqnd1} by setting $X = X'$ and $Y = Y'$).
If one can refute the conjecture, it will serve as evidence of $S(X;Y\|Z) \neq \hat{\Psi}(X;Y\|Z)$ for general distributions. If one can show the conjecture, one deduces the \emph{lower bound}:
$$
S(X';Y'\|Z') \geq S(X;Y\|Z) - S_{\text{ow}}(X,Y \rightarrow Z' \| Z) - I(Y',Z';X|X') - I(X',Z';Y|Y') - I(X;Y|X',Y',Z').
$$
This lower bound can be used to show that if
$$
S(X;Y\|Z) = I(X;Y)
$$
for some $P_{X,Y,Z}$ and if $X' = f_1(X)$, $Y' = f_2(Y)$ and $Z' = f_3(Z)$ for some functions $f_1, f_2, f_3$ such that $H(Z'|X',Y') = 0$, then $S(X';Y'\|Z') \geq I(X';Y')$ and hence,
$$
S(X';Y'\|Z') = I(X';Y').
$$
This shows that establishing the conjecture may require novel achievability (lower bound) ideas.

In Appendix \ref{appendixNew}, we show that Conjecture \ref{conjecture:1} is equivalent to an apparently weaker form of the conjecture.
}

\color{black}
\section{Is the upper bound $I(X;Y)$ tight for the XOR case?}\label{dis:tightness}
Let $X$ and $Y$ be binary random variables with an arbitrary joint distribution, and set $Z = X \oplus Y$. We conclude from Theorem \ref{th:binary} that
$$
S(X;Y\|Z) \leq \Psi(X;Y\|Z) = \hat{\Psi}(X;Y\|Z) = I(X;Y).
$$
In other words, the upper bound $I(X;Y)$ cannot be improved using any existing upper bound. Note that $I(X;Y)$ is the maximum rate we could have achieved if Eve did not have access to the sequence $Z^n$ and could only monitor the public channel. We conjecture that the rate $I(X;Y)$ is no longer achievable (for a general $P_{X,Y}$) when Eve has access to the XOR of $X$ and $Y$. 

\begin{conjecture}\label{conj:2}
    There exists some distribution $P_{X,Y}$ on binary random variables $X$ and $Y$ such that
    $$
    S(X;Y\|Z) < I(X;Y)
    $$
    where $Z = X \oplus Y$.
\end{conjecture}

One idea to develop a better upper bound is to note that for every $n$ we have
$$
S(X;Y\|Z) = \frac{1}{n} S(X^n; Y^n \| Z^n) \leq \frac{1}{n} \Psi(X^n; Y^n \| Z^n)
$$
where $(X^n, Y^n, Z^n)$ are $n$ i.i.d.\ repetitions of $(X, Y, Z)$. Therefore,
$$
S(X;Y\|Z) \leq \inf_n \left[ \frac{1}{n} \Psi(X^n; Y^n \| Z^n) \right].
$$

The above inequality might potentially lead to better upper bounds because $\Psi(X^n;Y^n\|Z^n)$ allows for minimizing over all $P_{J|X^n,Y^n,Z^n}$, which is a larger space than that of the single-letter bound. 
The following proposition shows that at least when $Z$ is a function of $(X,Y)$, we do not obtain better upper bounds with this approach, i.e., the proposition shows that
\begin{align}
\inf_n \left[\frac{1}{n} \Psi(X^n;Y^n\|Z^n)\right] = \Psi(X;Y\|Z) \label{eqnFW}
\end{align}
as $\Psi(X;Y\|Z) = \hat{\Psi}(X;Y\|Z)$ when $Z$ is a function of $(X,Y)$ (however, we do not know if \eqref{eqnFW} holds when $Z$ is not a function of $(X,Y)$).

\begin{proposition}\label{prop2}
    For every $n$, we have
    $$
    \hat{\Psi}(X;Y\|Z) = \frac{1}{n} \hat{\Psi}(X^n;Y^n\|Z^n)
    $$
where $(X^n, Y^n, Z^n)$ are $n$ i.i.d.\ repetitions of $(X, Y, Z)$.
\end{proposition}
\begin{proof}
Take some arbitrary $P_{J|X^n,Y^n,Z^n}$ and let $J_i = (J, X^{i-1}, Y^{i-1})$. Then, we have
\begin{align}
I(X^n;Y^n|J) + I(X^n,Y^n; J | Z^n) &= \sum_{i} \bigl( I(X_i; Y_i | J_i) + I(X^{i-1}; Y_i | J,Y^{i-1}) + I(Y^{i-1}; X_i | J,X^{i-1}) \bigr) \nonumber \\
&\quad + \sum_{i} \bigl( I(X_i,Y_i; J_i | Z_i) + I(X_i, Y_i; Z^{i-1}, Z_{i+1}^n | J_i,Z_i) \bigr) \nonumber \\
&\geq \sum_i \bigl( I(X_i; Y_i | J_i) + I(X_i,Y_i; J_i | Z_i) \bigr) \nonumber \\
&\geq n \cdot \hat{\Psi}(X;Y\|Z), \label{reverse:1}
\end{align}
where the last inequality follows from the definition of $\hat{\Psi}(X;Y\|Z)$.

For the converse direction, assume that $J_{*}$ achieves the minimum of $I(X;Y|J) + I(X,Y;J|Z)$.
Suppose $J$ consists of i.i.d.\ copies of the random variable $J_{*}$, i.e., $J = J_{*}^{n}$. We have 
\begin{align*}
I(X^n; Y^n | J_{*}^n) + I(X^n Y^n; J_{*}^n | Z^n) = n \bigl( I(X; Y | J_{*}) + I(X,Y; J_{*} | Z) \bigr) = n \hat{\Psi}(X;Y\|Z).
\end{align*}
Thus, $\hat{\Psi}(X^n; Y^n \| Z^n) \leq n \hat{\Psi}(X;Y\|Z)$. This observation and inequality \eqref{reverse:1} complete the proof. 
\end{proof}

\subsection{Controlling $H(Z^n|F_{1:r})$}
The following lemma shows that $H(Z^n|F_{1:r})$ cannot be equal to zero for any code that achieves a positive key rate:
\begin{proposition}
    \label{lmm3d}
    Take some $P_{X,Y}(x,y)$ such that $P_{X,Y}(x,y) > 0$ for all $x,y$. 
Assume that $H(Z^n|F_{1:r}) = 0$,
where $F_{1:r} = (F_1, F_2, \cdots, F_r)$ is the public communication. Then, 
$$
H(X^n, Y^n | F_{1:r}) = 0.
$$
\end{proposition} 

We need the following definition:
\begin{definition}\cite{ishwar,verdu} \label{defpreceq}
Given two pmfs $P_{X,Y}$ and $Q_{X,Y}$ on the alphabets $\mathcal{X} \times \mathcal{Y}$, the relation $Q_{X,Y} \preceq P_{X,Y}$ represents the existence of some functions $a: \mathcal{X} \mapsto \mathbb{R}_{+} \cup \{0\}$ and $b: \mathcal{Y} \mapsto \mathbb{R}_{+} \cup \{0\}$ such that for all $(x,y) \in \mathcal{X} \times \mathcal{Y}$ we have
\begin{align}
Q_{X,Y}(x,y) = a(x) b(y) P_{X,Y}(x,y). \label{eq:defforPart1proof}
\end{align}
\end{definition} 

\begin{proof}[Proof of Proposition \ref{lmm3d}]
Let us fix $F_{1:r} = f_{1:r}$. 
It is known that for any interactive communication,\footnote{This property is similar to the rectangle property of communication
complexity.} we have
$$
P_{X^n Y^n | F_{1:r} = f_{1:r}} \preceq P_{X^n Y^n}.
$$
In other words, 
$$
P(x^n, y^n | F_{1:r} = f_{1:r}) = P(x^n, y^n) a(x^n) b(y^n)
$$
for some functions $a(\cdot)$ and $b(\cdot)$.

Since $H(Z^n|F_{1:r} = f_{1:r}) = 0$, the value of $z^n$ must be {\color{black}fixed} when we fix $F_{1:r} = f_{1:r}$. This shows that for any $x_1^n, y_1^n$ and $x_2^n, y_2^n$ such that
$x_1^n + y_1^n \neq x_2^n + y_2^n$, the probabilities $p(x_1^n, y_1^n | F_{1:r} = f_{1:r})$ and $p(x_2^n, y_2^n | F_{1:r} = f_{1:r})$
cannot be positive at the same time. Thus, 
\[
0 = P(x_1^n, y_1^n | F_{1:r} = f_{1:r}) P(x_2^n, y_2^n | F_{1:r} = f_{1:r}) = P(x_1^n, y_1^n) P(x_2^n, y_2^n) a(x_1^n) b(y_1^n) a(x_2^n) b(y_2^n).
\]
Since $P_{X,Y}(x,y) > 0$ for all $x,y$, we must have
\[
a(x_1^n) b(y_1^n) a(x_2^n) b(y_2^n) = 0.
\]

Let $y^n$ be such that $b(y^n) > 0$. Take two distinct sequences $x_1^n, x_2^n$ and set $y_1^n = y_2^n = y^n$. This choice is valid since $x_1^n + y_1^n \neq x_2^n + y_2^n$. Thus, we get
\[
a(x_1^n) a(x_2^n) {\color{black} b^2(y^n)} = 0.
\]
Since we assumed $b(y^n) > 0$, we get that either $a(x_1^n)$ or $a(x_2^n)$ must be zero. Therefore, $a(x^n) = 0$ for all but one sequence $x^n$, and similarly we can show that $b(y^n) = 0$ for all but one sequence $y^n$. Thus, 
\[
P(x^n, y^n | F_{1:r} = f_{1:r}) = P(x^n, y^n) a(x^n) b(y^n)
\]
is positive for only one pair $(x^n, y^n)$. Thus,
\[
H(X^n, Y^n | F_{1:r} = f_{1:r}) = 0.
\]
\color{black}

\end{proof}

Proposition \ref{lmm3d} motivates the following definition: 
\begin{definition}
    $S_{\Delta}(X;Y\|Z)$ is the supremum of secret key rates corresponding to protocols
    satisfying  
    \[
    \lim_{n \rightarrow \infty} \frac{1}{n} H(Z^n | F_1, F_2, \cdots, F_r) \geq \Delta.
    \]
    The non-interactive capacity $S_{\text{ni},\Delta}(X;Y\|Z)$ is defined similarly.
\end{definition}

\begin{remark}
Observe that $\Delta \mapsto S_{\Delta}(X;Y\|Z)$ is a non-increasing function.
Moreover,
$$
S(X;Y\|Z) = \lim_{\Delta \rightarrow 0} S_{\Delta}(X;Y\|Z).
$$
\end{remark}

Our strategy for improving the upper bound on $S(X;Y\|Z)$ is as follows:
\begin{enumerate}
    \item Compute an upper bound for $S_{\Delta}(X;Y\|Z)$ for any given $\Delta > 0$. We will provide such an upper bound in Theorem \ref{tight:the1} below.
    \item Show that $S(X;Y\|Z) = S_{\Delta^*}(X;Y\|Z)$ for some $\Delta^* > 0$. Intuitively, we expect this to be true because of Proposition \ref{lmm3d}. We only show this part for the non-interactive setup for the XOR example.
    \item Use the upper bound on $S_{\Delta^*}(X;Y\|Z)$ to obtain an upper bound on $S(X;Y\|Z)$. We are unable to perform this last step for the non-interactive setting because, even though the upper bound on $S_{\Delta^*}(X;Y\|Z)$ is theoretically computable, unfortunately, the number of free variables is large and it is not easy to obtain reliable numerical results.
\end{enumerate}

We begin with task 1:

\begin{theorem}\label{tight:the1}
    $S_{\Delta}(X;Y\|Z)$ has an upper bound as follows:
    \begin{align}
       S_{\Delta}(X;Y\|Z) \leq \Psi_{\Delta}(X;Y\|Z) \triangleq \inf_J \max_{(U,V) \mkv (X,Y) \mkv (J,Z)} I(X;Y|J) + I(U;J|V) - I(U;Z|V), \label{main:delta1}
    \end{align}
    where $P_{X,Y,Z,J,U,V}$ satisfies the following constraint:
    \begin{align}
      I(V;Z) \leq I(V;J) + H(Z) - \Delta, \label{cons1}
    \end{align}
    and in the inner maximum, $U$ and $V$ satisfy the cardinality bounds 
    $|\mathcal{V}| \leq |\mathcal{X}||\mathcal{Y}| + 1$ and $|\mathcal{U}| \leq |\mathcal{X}||\mathcal{Y}|$.
    
    Moreover, if $Z = f(X,Y)$, in computing the upper bound we can assume that $H(X,Y | U,V,Z) = 0$.
\end{theorem}

\begin{remark}\label{re:2} 
Observe that $\Delta \mapsto \Psi_{\Delta}(X;Y\|Z)$ is a non-increasing function. When $\Delta = 0$, $\Psi_{\Delta}(X;Y\|Z)$ reduces to $\Psi(X;Y\|Z)$ as the only difference between the optimization problems in \eqref{main1} and \eqref{main:delta1} is the additional constraint defined in \eqref{cons1}. This constraint holds trivially when $\Delta = 0$. This constraint imposes a limitation on the distribution $P_{U,V|X,Y}$. 

As observed earlier, in computing $\Psi(X;Y\|Z)$ when $Z$ is a function of $(X,Y)$, an optimizer for 
$$
\max_{(U,V) \mkv (X,Y) \mkv (J,Z)} I(U;J|V) - I(U;Z|V)
$$
is $U = (X,Y)$ and $V = Z$, which yields the value $I(X,Y;J|Z)$. However, for a positive 
$\Delta$, this choice of $(U,V)$ may not be feasible in \eqref{cons1}.
{\color{black}
The constraint 
\begin{align*}
    I(V;Z) \leq I(V;J) + H(Z) - \Delta
\end{align*}
for the choice of $V=Z$ reduces to
\begin{align*}
    \Delta \leq I(Z;J).
\end{align*}
In particular, when $J$ is nearly a constant random variable, $I(Z;J) \approx 0$, which implies that $\Delta$ must be close to zero for $V = Z$ to be admissible. The second task is to prevent $\Delta$ from being close to zero.

}

\end{remark}

Proof of Theorem \ref{tight:the1} is given in Appendix \ref{append2}.

We now turn to the second task. The following theorem provides a lower bound on $H(Z^n | F_1, F_2)$ for a capacity-achieving code. Some of the manipulations inside the proof are similar to those given in \cite{nair2020} in the context of a different distributed source coding problem (the modulo-two sum problem of Körner and Marton).
\begin{theorem}\label{sd=s}
Suppose that $X$ and $Y$ are binary random variables, and let 
$Z = X \oplus Y$. Define
\begin{align}
  \Delta_1 &= \min_{U \mkv X \mkv (Y,Z)} \bigl[ H(Z|U) - H(X|U) \bigr],~~
  \Delta_2 = \min_{V \mkv Y \mkv (X,Z)} \bigl[ H(Z|V) - H(X|V) \bigr], \\
  \bar{\Delta} &= \max\{\Delta_1, \Delta_2\}.
\end{align}
Then, every non-interactive code with public messages $F_1$ and $F_2$ satisfies
\begin{align}
  \frac{1}{n} H(Z^n | F_1, F_2) \geq \frac{\bar{\Delta}}{2} + \frac{I(K_A; K_B | F_1, F_2)}{2n}.
\end{align}
\end{theorem}

\begin{corollary} 
Let 
\[
\Delta = \frac{\bar{\Delta} + S_{\mathrm{ni}}(X;Y\|Z)}{2}.
\]
Then,
\begin{align}
  S_{\mathrm{ni}}(X;Y\|Z) = S_{\mathrm{ni}, \Delta}(X;Y\|Z).
\end{align}
\end{corollary}

\begin{proof}
We have
\begin{align*}
  I(K_A; K_B | F_1, F_2) &\leq I(X^n; Y^n | F_1, F_2) \\
  &\leq H(X^n, Y^n | F_1, F_2) \\
  &= n H(X,Y) - I(F_1, F_2; X^n, Y^n).
\end{align*}
This implies
\[
  \frac{1}{n} I(F_1, F_2; X^n, Y^n) \leq H(X,Y) - \frac{1}{n} I(K_A; K_B | F_1, F_2).
\]

To prove the desired inequality, it suffices to show the following bounds:
\begin{align}
  \frac{1}{n} I(F_1, F_2; X^n, Y^n) + \frac{2}{n} H(Z^n | F_1, F_2) &\geq 
  H(X,Y) + \min_{U \mkv X \mkv (Y,Z)} [ H(Z|U) - H(Y|U) ], \label{ineq:11} \\
  \frac{1}{n} I(F_1, F_2; X^n, Y^n) + \frac{2}{n} H(Z^n | F_1, F_2) &\geq 
  H(X,Y) + \min_{V \mkv Y \mkv (X,Z)} [ H(Z|V) - H(X|V) ]. \label{ineq:12}
\end{align}

Combining \eqref{ineq:11} and \eqref{ineq:12} yields
\begin{align}
  \frac{1}{n} I(F_1, F_2; X^n, Y^n) + \frac{2}{n} H(Z^n | F_1, F_2) \geq H(X,Y) + \bar{\Delta}. \label{ineq:13}
\end{align}

We only prove inequality \eqref{ineq:11} since \eqref{ineq:12} is analogous.

Let 
$U_i = (F_1, Y^{i-1}, Z_{i+1}^n), \quad V_i = (F_2, X^{i-1}, Z_{i+1}^n)$ be new auxiliary random variables. We have:
\begin{align}
  & I(F_1, F_2; X^n, Y^n) + 2 H(Z^n | F_1, F_2) \nonumber \\
  &= I(F_1; X^n, Y^n) + I(F_2; Y^n, X^n | F_1) + 2 H(Z^n | F_1, F_2) \nonumber \\
  &= I(F_1; X^n) + I(F_2; Y^n, Z^n | F_1) + 2 H(Z^n | F_1, F_2) \nonumber \\
  &= I(F_1; X^n) + I(F_2; Z^n | F_1) + I(F_2; Y^n | F_1, Z^n) + 2 H(Z^n | F_1, F_2) \nonumber \\
  &= I(F_1; X^n) + H(Z^n | F_1) + I(F_2; Y^n | F_1, Z^n) + H(Z^n | F_1, F_2) \nonumber \\
  &\overset{(a)}{=} I(F_1; X^n) + H(Y^n | F_1) + H(Z^n | F_1, Y^n) + I(F_2; Y^n | F_1, Z^n) \nonumber \\
  &\quad + H(Z^n | F_1) - H(Y^n | F_1) + I(Y^n; Z^n | F_1, F_2) \nonumber \\
  &\overset{(b)}{=} H(X^n, Y^n) + I(F_2; Y^n | F_1, Z^n) + H(Z^n | F_1) - H(Y^n | F_1) + I(Z^n; Y^n | F_1, F_2) \nonumber \\
  &\overset{(c)}{\geq} n H(X,Y) + H(Z^n | F_1) - H(Y^n | F_1) \nonumber \\
  &\overset{(d)}{=} n H(X,Y) + \sum_i H(Z_i | F_1, Y^{i-1}, Z_{i+1}^n) - H(Y_i | F_1, Y^{i-1}, Z_{i+1}^n) \nonumber \\
  &= n H(X,Y) + \sum_i \left[ H(Z_i | U_i) - H(Y_i | U_i) \right], \label{ineq:inq15}
\end{align}
where $(a)$ follows from $H(Z^n | F_1, F_2, Y^n) = H(Z^n | F_1, Y^n)$, which is equivalent to $H(X^n | F_1, F_2, Y^n) = H(X^n | F_1, Y^n)$ which holds due to the Markov chain  chain $F_2, Y^n \mkv X^n \mkv F_1$, and for 
  $(b)$ we use the following
 \begin{align*}
     I(F_1; X^n) + H(Y^n | F_1) + H(Z^n | F_1, Y^n) &=I(F_1;X^n)+H(Y^n|F_1)+H(X^n|F_1,Y^n)\\
     &=I(F_1;X^n)+H(X^n,Y^n|F_1)\\
     &=H(X^n, Y^n).  
 \end{align*}

  Also, in the above, $(c)$ follows from  the non-negativity of conditional  mutual information and 
  $(d)$ follows from the Marton-Körner identity.

Let \( Q \) be a random variable uniformly distributed over the set \( \{1, 2, \ldots, n\} \), and let it be independent of the variables \( X \), \( Y \), and \( F_{1:r} \). Let 
$X = X_Q, \quad Y = Y_Q$ and $U = (U_Q, Q)$.
It is clear that we have the Markov chain $U \mkv X \mkv Y$. Thus, \eqref{ineq:inq15} simplifies to the desired bound.
\end{proof}

\begin{remark}
We note that the optimization problems
$
\min_{U \mkv X \mkv Y} H(Z|U) - H(Y|U) \quad \text{and} \quad \min_{V \mkv Y \mkv X} H(Z|V) - H(X|V)
$
also arise in other contexts, including the "modulo 2 sum" problem studied in \cite{nair2020}. For instance, the following is known:
\end{remark}

\begin{example}[Propositions 1 and 2 of \cite{nair2020}]
Let $P_X(0) = x$, $P_{X|Y}(0|0) = c$, and $P_{X|Y}(1|1) = d$. Then, $\bar{\Delta} = 0$ if $(1 - 2c)(1 - 2d) \leq 0$. Moreover, by choosing
\[
x = \frac{\sqrt{d(1 - d)}}{\sqrt{d(1 - d)} + \sqrt{c(1 - c)}},
\]
the value of $\bar{\Delta}$ is
$
2 H(Z) - H(X,Y)
$
if either $c = d$ or $(1 - 2c)(1 - 2d) > 0$ and $c \neq d$.
\end{example}

\section*{Acknowledgment}

Prior to the involvement of the first author in this problem, the second author had designated this project as the undergraduate thesis topic for Yicheng Cui. During the initial stages of the research, Cui made some progress on the XOR and AND patterns, using different (and generally weaker) techniques than those presented in this work.

The authors also thank Mr.\ Chin Wa Lau for his valuable assistance in proving Proposition \ref{prop2} and for supporting Yicheng Cui during his undergraduate thesis.

\appendix
\section{Capacity of a Class of Source Distributions}
\label{appndx}

Assume that $Z = g(X)$ and the channel from $X$ to $Y$ is an erasure channel, i.e., $Y = X$ with probability $1 - \epsilon$ and $\{e\}$ otherwise. 
Note that the random variables $X$, $Y$, and $Z$ do not necessarily form a Markov chain in any order, and the key capacity for this class is not known in the literature. 

We claim that $S(X;Y\|Z) = I(X;Y|Z) = (1 - \epsilon) H(X|Z)$ is achievable by one-way communication from Alice to Bob. Note that 
\begin{align}
&S_{\text{ow}}(X \rightarrow Y\|Z) = \max_{(U,V) \mkv X \mkv (Y,Z)} I(U;Y|V) - I(U;Z|V).
\end{align}
Consider $V$ as constant, and by the functional representation lemma, let $U$ be such that $I(Z;U) = 0$ and $H(X|Z,U) = 0$. Then, note that $I(U;Z|Y) = 0$ since, when $Y$ is $X$, it holds, and when $Y$ is $e$, it also holds. Then, we have
\begin{align}
I(U;Y|V) - I(U;Z|V) &= I(U;Y) - I(U;Z)\\
&= I(U;Y)\\
&= I(U;Y,Z)\\
&= I(U;Y|Z)\\
&= I(X;Y|Z) - I(X;Y|Z,U)\\
&= I(X;Y|Z).
\end{align}
This completes the proof.

{\color{black} The significance of this example lies in  demonstrating that \( S(X;Y\|Z) \) equals the one-way communication rate, despite \( X, Y, Z \) not forming a Markov chain in any arrangement.}
\section{Some of the Proofs}\label{append2}
\subsection{Proof of Lemma \ref{lem:con-hull}}
\begin{proof}[Proof of Lemma \ref{lem:con-hull}]
Note that $I(Z;T|X,Y) = 0$ holds trivially when $Z$ is a function of $(X,Y)$. Assume that, for a given $P_{X,Y}$, there exists a random variable $T$ such that $I(Z;T|X,Y) = I(Z;T) = I(X;Y|T) = 0$. Define $Q_{X,Y}^{(t)} = P(X,Y|T = t)$. From $I(X;Y|T) = 0$, it follows that $I_{Q^{(t)}}(X;Y) = 0$. Additionally, for any $t \in \mathcal{T}$ and $z \in \mathcal{Z}$, we have 
\begin{align*}
    Q^{(t)}(Z = z) &= \sum_{x,y \in \mathcal{R}_{z}} Q^{(t)}(X = x, Y = y) \\
    &= \sum_{x,y \in \mathcal{R}_{z}} P(X = x, Y = y | T = t) \\
    &= P(Z = z | T = t) \overset{(a)}{=} P(Z = z),
\end{align*}
where $\mathcal{R}_{z} = \left\{(x,y) \in \mathcal{X} \times \mathcal{Y} : z = f(x,y) \right\}$, and $(a)$ follows from $I(Z;T) = 0$. Hence, we conclude that $P_{X,Y} \in \mathrm{conv}(\color{black}{\mathcal{A}})$.
 
For the converse direction, assume that $P_{X,Y} \in \mathrm{conv}(\color{black}{\mathcal{A}})$. Therefore, we can represent $P_{X,Y}$ as a convex combination of elements from $\color{black}{\mathcal{A}}$:
\begin{align}
     P_{X,Y} = \sum_{t} \omega_{t} Q_{X,Y}^{(t)}, \quad Q_{X,Y}^{(t)} \in \color{black}{\mathcal{A}},
\end{align}
where $\omega_t \geq 0$ and $\sum_{t} \omega_{t} = 1$.
Define $P(X,Y|T = t) = Q_{X,Y}^{(t)}$. Since $Q^{(t)} \in \color{black}{\mathcal{A}}$, it follows that $I_{P}(X;Y) = I_{Q^{(t)}}(X;Y) = 0$. Next, 
\begin{align*}
    P(Z = z | T = t) &= \sum_{x,y \in \mathcal{R}_{z}} P(X = x, Y = y | T = t) \\
    &= \sum_{x,y \in \mathcal{R}_{z}} Q^{(t)}(X = x, Y = y) \\
    &= Q(Z = z) \overset{(a)}{=} P(Z = z),
\end{align*}
where $(a)$ holds since $Q^{(t)} \in \color{black}{\mathcal{A}}$. Thus, from above, we have $I_{P}(Z;T) = 0$. This completes the proof.  
\end{proof}
\subsection{Proof of Theorem \ref{th:binary}}
\begin{proof}[Proof of Theorem \ref{th:binary}]
The condition $\mathrm{Cov}(X,Y) \leq 0$ is equivalent to $p_{00}p_{11} \leq p_{01}p_{10}$. 

{\color{black} Consider the first part: we claim that
\begin{align}
\Psi(X;Y\|Z_3) \geq \Psi(X;Y\|Z_1).\label{mmmeq1}
\end{align}
When $Z$ is a function of $(X,Y)$ we can write
\begin{align*}
    \Psi(X;Y\|Z)
    &= \min_{J} I(X;Y|J) + I(X,Y;J|Z) \\
    &= \min_{J} I(X;Y|J) + I(X,Y;J) - I(J;Z).
\end{align*}
Since $Z_3$ is a function of $Z_1$, we get \eqref{mmmeq1}. 
}
Since
\[
\Psi(X;Y\|Z_3) \geq \Psi(X;Y\|Z_1),
\]
it suffices to show that, if $p_{00}p_{11} \leq p_{01}p_{10}$, then $\Psi(X;Y\|Z_1) = I(X;Y)$, and if $p_{00}p_{11} > p_{01}p_{10}$, then $\Psi(X;Y\|Z_3) = 0$.

\medskip
\noindent\textbf{Proof of $\Psi(X;Y\|Z_1) = I(X;Y)$ when $p_{00}p_{11} \leq p_{01}p_{10}$:}  
The condition $p_{00}p_{11} \leq p_{01}p_{10}$ implies that the polynomial
\[
x^2 - (p_{01} + p_{10})x + p_{00}p_{11} = 0
\]
has two non-negative real roots. Let $x_0 \leq x_1$ denote these roots. Note that $x_0x_1 = p_{00}p_{11}$ and $x_0 + x_1 = p_{01} + p_{10}$. One can verify that $p_{01} \in [x_0, x_1]$. Then, there exists $\lambda \in [0,1]$ such that
\[
p_{01} = \lambda x_0 + (1 - \lambda) x_1.
\]
Since $x_0 + x_1 = p_{01} + p_{10}$, we obtain
\[
p_{10} = \lambda x_1 + (1 - \lambda) x_0.
\]

Since $\Psi(X;Y\|Z_1) = \hat{\Psi}(X;Y\|Z_1)$, to show that $\hat{\Psi}(X;Y\|Z_1) = I(X;Y)$, we use Theorem \ref{Th:gen:classic}. It suffices to identify a $T$ such that $I(T;Z_1|X,Y) = I(T;Z_1) = I(X;Y|T) = 0$. The condition $I(T;Z_1|X,Y) = 0$ holds for any $P_{T|X,Y}$ because $Z_1$ is a function of $(X,Y)$. 

Let $T$ be a binary random variable satisfying $P(T = 1) = \lambda$ and $P(T = 0) = 1 - \lambda$, and
\[
P(X = 1, Y = 1 | T = t) = p_{11}, \quad t \in \{0,1\},
\]
\[
P(X = 0, Y = 0 | T = t) = p_{00}, \quad t \in \{0,1\},
\]
\[
P(X = 0, Y = 1 | T = 0) = P(X = 1, Y = 0 | T = 1) = x_1,
\]
\[
P(X = 0, Y = 1 | T = 1) = P(X = 1, Y = 0 | T = 0) = x_0.
\]
This yields the desired result.

\medskip
\noindent\textbf{Proof of $\Psi(X;Y\|Z_3) = 0$ when $p_{00}p_{11} > p_{01}p_{10}$:}  
We identify a binary random variable $J^*$ such that $I(X;Y|J^*) = I(X,Y;J^*|Z_3) = 0$. This would imply that
\[
\Psi(X;Y\|Z_3) \leq I(X;Y|J^*) + I(X,Y;J^*|Z_3) = 0.
\]
Let
\[
\alpha = \frac{p_{00}}{p_{00} - (p_{00}p_{11} - p_{01}p_{10})}.
\]
Under the condition $p_{00}p_{11} - p_{01}p_{10} > 0$, we have $\alpha \geq 1$. Suppose that
\[
P_{X,Y|J^* = 0} =
\begin{bmatrix}
0 & 0 \\
0 & 1
\end{bmatrix},
\quad
P_{X,Y|J^* = 1} =
\begin{bmatrix}
\alpha p_{00} & \alpha p_{01} \\
\alpha p_{10} & 1 - \alpha(1 - p_{11})
\end{bmatrix},
\]
and $P(J^* = 1) = \frac{1}{\alpha}$. One can directly verify that
\[
P_{X,Y} = P(J^* = 1) P_{X,Y|J^* = 1} + P(J^* = 0) P_{X,Y|J^* = 0}.
\]
Moreover, for this choice of $J^*$, we have $I(X;Y|J^*) = I(X,Y;J^*|Z) = 0$.

\medskip
\noindent\textbf{Second part:}  
We consider two cases. If $p_{00}p_{11} \leq p_{01}p_{10}$, we use the fact that $\Psi(X;Y\|Z_1) = I(X;Y)$ and
\[
\Psi(X;Y\|Z_2) \geq \Psi(X;Y\|Z_1) = I(X;Y).
\]
In the case of $p_{00}p_{11} \leq p_{01}p_{10}$, consider the following random variable:
\[
\hat{Z}_1 =
\begin{cases}
0 & (X,Y) = (0,1), \\
1 & (X,Y) = (1,0), \\
2 & (X,Y) \in \{(0,0), (1,1)\}.
\end{cases}
\]
By symmetry, from the first part we deduce that, in the case $p_{00}p_{11} \leq p_{01}p_{10}$, we have $\Psi(X;Y\|\hat{Z}_1) = I(X;Y)$. Since
\[
\Psi(X;Y\|Z_2) \geq \Psi(X;Y\|\hat{Z}_1) = I(X;Y),
\]
we obtain the desired result.
\end{proof}
\subsection{Proof of Theorem \ref{th:ter-bin}}
\begin{proof}[Proof of Theorem \ref{th:ter-bin}]
Note that when $P(Y=2)=0$, i.e., when $Y\in\{0,1\}$, the random variable $Z$ can be expressed as $Z=X\oplus Y$. In this case, the proof of Theorem 
\ref{th:binary} shows that one can find a random variable $T$ such that $I(Z;T|X,Y)=I(Z;T)=I(X;Y|T)=0$. Lemma \ref{lem:con-hull} would then imply that any $P_{X,Y}$ satisfying $P_{Y}(2)=0$ belongs to the convex hull of the set $\mathcal{A}(P_Z)$ as defined in \eqref{con:formula}. More specifically, let 
$\alpha=P_Z(0)$ and $\bar{\alpha}=P_Z(1)=1-\alpha$, and let us denote $\mathcal{A}(P_Z)$ by $\mathcal{A}(\alpha)$. Then, for any $r,t\in[0,1]$, the following joint distribution belongs to 
the convex hull of the set $\mathcal{A}(\alpha)$:
\begin{align}
  B_{r,t}(x,y)=\begin{tabular}{|c|c|c|c|}
 \hline
   $\mathbf{X {\color{black}\resizebox{6mm}{5mm}{$\backslash$}} Y}$& $\mathbf{0}$ & $\mathbf{1}$ & $\mathbf{2}$   \\
  \hline
 $\mathbf{0}$&$\alpha r$& $\bar{\alpha}\bar{t}$ &0\\
\hline
 $\mathbf{1}$&$\bar{\alpha}t$ & $\alpha\bar{r}$ &0 \\ 
\hline
\end{tabular}
 \end{align}
{\color{black}
In the above table, the values of $X$ are listed in the first column (i.e., $X \in \{0, 1\}$) and the values of $Y$ are listed in the first row (i.e., $Y \in \{0, 1, 2\}$). The remaining cells of the table contain probabilities.}

Observe that when $P(Y=0)=0$, the mapping from $X$ and $Y\in\{1,2\}$ to $Z$ is again an XOR map. A similar argument then shows that any $P_{X,Y}$ satisfying $P_{Y}(0)=0$ belongs to the convex hull of the set $\mathcal{A}(\alpha)$.
In other words, for any $s,\ell\in[0,1]$, the following joint distribution belongs to the convex hull of the set $\mathcal{A}(\alpha)$:
\begin{align}
 C_{s,\ell}(x,y)=\begin{tabular}{|c|c|c|c|}
 \hline
   $\mathbf{X {\color{black}\resizebox{6mm}{5mm}{$\backslash$}} Y}$& $\mathbf{0}$ & $\mathbf{1}$ & $\mathbf{2}$  \\
 \hline
$\mathbf{0}$& 0& $\bar{\alpha}\bar{s}$ & $\alpha \ell$  \\
\hline
$\mathbf{1}$&0& 
$\alpha\bar{\ell}$ & $\bar{\alpha}s$ \\ 
\hline
\end{tabular}
 \end{align}
Observe that for any $\omega\in[0,1]$, the linear combination
$$\omega B_{r,t}(x,y)+\bar{\omega}C_{s,\ell}(x,y)$$
must also belong to the convex hull of the set $\mathcal{A}(\alpha)$.

Let $p_{ij}=P_{X,Y}(i,j)$. We have $\alpha=p_{00}+p_{02}+p_{11}$. Set 
$$r=\min\left\{1, \frac{\bar{\alpha}p_{00}}{\alpha p_{10}} \right\},$$
$$\omega=\frac{ p_{00}}{\alpha r}, \quad t=\frac{\alpha p_{10}}{\bar{\alpha}p_{00}}r, \quad 
s=\frac{\alpha p_{12}r}{\bar{\alpha}(\alpha r-p_{00})}, \quad
\ell=\frac{p_{02}r}{\alpha r-p_{00}}.$$
One can then verify that
$$P_{X,Y}(x,y)=\omega B_{r,t}(x,y)+\bar{\omega}C_{s,\ell}(x,y)$$
belongs to the convex hull of the set $\mathcal{A}(\alpha)$.
The conditions $$\frac{p_{00}}{\alpha}+\frac{p_{12}}{\bar{\alpha}}\leq 1, \quad \frac{p_{02}}{\alpha}+\frac{p_{10}}{\bar{\alpha}}\leq 1$$  
imply that 
$$\omega,t,r,s,\ell\in[0,1].$$ Consequently, one can identify a random variable $T$ such that $I(Z;T|X,Y)=I(Z;T)=I(X;Y|T)=0$ holds for $P_{X,Y}$, provided that
the conditions $$\frac{p_{00}}{\alpha}+\frac{p_{12}}{\bar{\alpha}}\leq 1, \quad \frac{p_{02}}{\alpha}+\frac{p_{10}}{\bar{\alpha}}\leq 1$$  
are satisfied.

\end{proof}

\subsection{Proof of Theorem \ref{thm:thm5}}
\begin{proof}[Proof of Theorem \ref{thm:thm5}]
Let the alphabets of $X$ and $Y$ take values in $\{0,1,2,\cdots,s_{1}-1\}$ and $\{0,1,2,\cdots,s_{2}-1\}$, respectively. Let $Z=f(X,Y)$.  
The function $f$ corresponds to a table
\begin{align}
  \begin{tabular}{|c|c|c|c|c|}
 \hline
  $\mathbf{X {\color{black}\resizebox{6mm}{5mm}{$\backslash$}}Y}$ & $\mathbf{0}$ & $\mathbf{1}$ & $\cdots$ & $\mathbf{s_2-1}$ \\
  \hline
  $\mathbf{0}$ &  $f(0,0)$ &  $f(0,1)$ & $\cdots$ &  $f(0,s_2-1)$  \\
  $\mathbf{1}$ &  $f(1,0)$ &  $f(1,1)$ &  $\cdots$ & $f(1,s_2-1)$ \\
  $\vdots$ &  $\vdots$ &  $\vdots$ &  $\cdots$ & $\vdots$  \\
  $\mathbf{s_1-1}$ &  $f(s_1-1,0)$ & $f(s_1-1,1)$ & $\cdots$ & $f(s_1-1,s_2-1)$  \\
  \hline
\end{tabular}
\end{align}

{\color{black}
In the above table, the values of $X$ are listed in the first column (i.e., $X \in \{0, 1, \dots\}$) and the values of $Y$ are listed in the first row (i.e., $Y \in \{0, 1, \dots\}$). The remaining cells of the table contain the values $f(x,y)$. For example, the entry in the second row and second column is $f(0,0)=f(X=0,Y=0)$.}

A table is considered valid if the equality $\Psi(X;Y|Z)=I(X;Y)$ holds for all distributions $P_{X,Y}$. Since any distribution over a subset of $\mathcal{X} \times \mathcal{Y}$ can be extended to a distribution over $\mathcal{X} \times \mathcal{Y}$, it follows that if a table is valid, all of its sub-tables will also be valid. Consider a $2\times 2$ table.  
As stated in Theorem \ref{th:binary}, the equality $\Psi(X;Y|Z)=I(X;Y)$ does not hold for any arbitrary distribution $P_{X,Y}$ when the functions are AND, OR, or $Z_1$. It is also evident that if $Z=f(X,Y)$ is a one-to-one map, then $\Psi(X;Y|Z)\leq I(X;Y|Z)=0$, and the equality $\Psi(X;Y|Z)=I(X;Y)$ would not hold for every $P_{X,Y}$. Consequently, for $2\times 2$ tables, the following combinations are invalid if $A,B,C$, and $D$ are distinct symbols:
 \begin{align}
\begin{tabular}{|c|c|c|}
  \hline
   $\mathbf{X {\color{black}\resizebox{6mm}{5mm}{$\backslash$}} Y}$ & $\mathbf{0}$ & $\mathbf{1}$  \\
  \hline
  $\mathbf{0}$ &  $A$ &  $B$  \\
  $\mathbf{1}$ &  $B$ &  $C$ \\  
  \hline
\end{tabular}
~~
\begin{tabular}{|c|c|c|}
  \hline
 $\mathbf{X {\color{black}\resizebox{6mm}{5mm}{$\backslash$}} Y}$ & $\mathbf{0}$ & $\mathbf{1}$  \\
  \hline
  $\mathbf{0}$ &  $A$ &  $B$  \\
  $\mathbf{1}$ &  $C$ &  $A$ \\  
  \hline
\end{tabular}
~~
\begin{tabular}{|c|c|c|}
  \hline
   $\mathbf{X {\color{black}\resizebox{6mm}{5mm}{$\backslash$}} Y}$ & $\mathbf{0}$ & $\mathbf{1}$   \\
  \hline
  $\mathbf{0}$ &  $A$ &  $B$  \\
  $\mathbf{1}$ &  $A$ &  $A$ \\  
  \hline
\end{tabular}
~~
\begin{tabular}{|c|c|c|}
  \hline
   $\mathbf{X {\color{black}\resizebox{6mm}{5mm}{$\backslash$}} Y}$ & $\mathbf{0}$ & $\mathbf{1}$  \\
  \hline
  $\mathbf{0}$ &  $A$ &  $B$  \\
  $\mathbf{1}$ &  $C$ &  $D$ \\  
  \hline
  \end{tabular}
\end{align}
and the valid $2 \times 2$ tables are limited to the following:
\begin{align}\label{valid:2*2:tables}
 T_1=\begin{tabular}{|c|c|c|}
  \hline
   $\mathbf{X {\color{black}\resizebox{6mm}{5mm}{$\backslash$}} Y}$ & $\mathbf{0}$ & $\mathbf{1}$  \\
  \hline
  $\mathbf{0}$ &  $A$ &  $B$  \\
  $\mathbf{1}$ &  $B$ &  $A$ \\  
  \hline
\end{tabular}
\qquad\qquad T_2=\begin{tabular}{|c|c|c|}
  \hline
   $\mathbf{X {\color{black}\resizebox{6mm}{5mm}{$\backslash$}} Y}$ & $\mathbf{0}$ & $\mathbf{1}$ \\
  \hline
  $\mathbf{0}$ &  $A$ &  $A$  \\
  $\mathbf{1}$ &  $A$ &  $A$ \\  
  \hline
\end{tabular}.
\end{align}

Now, consider a general table $T$ of arbitrary size $s_1\times s_2$ where $(s_1,s_2)\neq (2,2)$. Since every $2\times 2$ sub-table must be valid, we conclude that for every $i_1\neq i_2$ and $j_1\neq j_2$, we must either have 
$$T(i_1,j_1) = T(i_1,j_2) = T(i_2,j_1) = T(i_2,j_2)$$
or
$$T(i_1,j_1) = T(i_2,j_2) \neq T(i_2,j_1) = T(i_1,j_2).$$
In both cases, $T(i_1,j_1) = T(i_2,j_2)$ for every $i_1\neq i_2$ and $j_1\neq j_2$. In other words, if we move from one cell to another cell by changing both the row and the column, the value of the new cell must be equal to the value of the old cell. Note that in an $s_1\times s_2$ table where $(s_1,s_2)\neq (2,2)$, we can move from a cell to all the other cells by a sequence of such moves. For instance, in a $2\times 3$ table, we can move from $(1,1)$ to $(1,2)$ by the following sequence:
$$(1,1)\rightarrow (2,2)\rightarrow (1,3)\rightarrow (1,2).$$
This shows that when $(s_1,s_2)\neq (2,2)$, $T(i_1,j_1) = T(i_2,j_2)$ for all $i_1,i_2,j_1,j_2$. 
\end{proof}

\subsection{Proof of Theorem \ref{thm6n}}
\begin{proof}[Proof of Theorem \ref{thm6n}]
{\color{black}Let $P_{J|X',Y',Z'}$ be a minimizer for the optimization problem in $\hat{\Psi}(X',Y'\|Z')$, i.e., $$\hat{\Psi}(X',Y'\|Z')=I(X';Y'|J)+I(X',Y';J|Z').$$ 
We know that such a minimizer exists because a cardinality bound on $J$ is available for $\hat{\Psi}(X',Y'\|Z')$.

We define the joint distribution between $J$ and $(X,Y,Z,X',Y',Z')$ as 
$P_{J|X',Y',Z'}P_{X',Y',Z',X,Y,Z}$.
This construction implies the Markov chain
\begin{align}
 (X,Y,Z)\mkv (X',Y',Z')\mkv J.  \label{mar:J} 
\end{align}

Since
$$\hat{\Psi}(X;Y\|Z)\leq I(X;Y|J)+I(X,Y;J|Z),$$
it suffices to show that
\begin{align}
I(X;Y|J)+I(X,Y;J|Z)-\hat{\Psi}(X';Y'\|Z')&=
I(X;Y|J)+I(X,Y;J|Z)-I(X';Y'|J)-I(X',Y';J|Z')\nonumber
\\&\leq I(X,Y;Z'|Z)+I(Y',Z';X|X')+I(X',Z';Y|Y')+I(X;Y|X',Y',Z').\label{inq:ak1}
\end{align}
This inequality follows from adding the following three inequalities:
\begin{align}
I(X;Y|J)+I(X,Y;J|Z)-I(X,X';Y,Y'|J)-I(X,X',Y,Y';J|Z')
&\leq I(X,Y;Z'|Z),\label{eqnTT1}\\
I(X,X';Y,Y'|J)+I(X,X',Y,Y';J|Z')-
I(X,X';Y'|J)-I(X,X',Y';J|Z')&\leq I(X,X',Z';Y|Y'),\label{eqnTT2}\\
I(X,X';Y'|J)+I(X,X',Y';J|Z')-I(X';Y'|J)-I(X',Y';J|Z')&\leq I(Y',Z';X|X').\label{eqnTT3}
\end{align}

The inequality \eqref{eqnTT1} follows from
\begin{align}
&I(X;Y|J)+I(X,Y;J|Z)-I(X,X';Y,Y'|J)-I(X,X',Y,Y';J|Z')\nonumber
\\&\leq I(X,Y;J,Z'|Z)-I(X,X',Y,Y';J|Z')\nonumber
\\&\leq I(X,Y;Z'|Z)+I(X,Y,Z;J|Z')-I(X,X',Y,Y';J|Z')\nonumber
\\&\leq I(X,Y;Z'|Z),\nonumber
\end{align}
where the last step follows from \eqref{mar:J}.

The inequality \eqref{eqnTT2} follows from
\begin{align}
&I(X,X';Y,Y'|J)+I(X,X',Y,Y';J|Z')-
I(X,X';Y'|J)-I(X,X',Y';J|Z')\nonumber\\&
=I(X,X';Y|J,Y')+I(Y;J|X,X',Y',Z')
\nonumber\\&
=I(X,X';Y|J,Y')-I(Y;J|X,X',Y',Z') \label{stepus}\\
&\leq I(X,X',Z';Y|Y')-
I(J;Y|Y')\nonumber
\\&\leq I(X,X',Z';Y|Y'),\nonumber
\end{align}
where \eqref{stepus} follows from \eqref{mar:J}.

Finally, the inequality \eqref{eqnTT3} follows similarly from
\begin{align}
&I(X,X';Y'|J)+I(X,X',Y';J|Z')-I(X';Y'|J)-I(X',Y';J|Z')\nonumber\\
&=I(X;Y'|X',J)+I(X;J|X',Y',Z')\nonumber\\
&=I(X;Y'|X',J)-I(X;J|X',Y',Z')\nonumber\\
&\leq I(Y',Z';X|X')-I(J;X|X')\nonumber\\
&\leq I(Y',Z';X|X').\nonumber
\end{align}
This completes the proof.}
\end{proof}

\subsection{Proof of Theorem \ref{tight:the1}}
\begin{proof}[Proof of Theorem \ref{tight:the1}]
Let $F_{1:r}=(F_1,\cdots,F_r)$. Take some arbitrary $P_{J|X,Y,Z}$ and suppose $(X^n,Y^n,Z^n,J^n)$ are $n$ i.i.d.\ copies according to $P_{X,Y,Z}P_{J|X,Y,Z}$. 
Define the auxiliary random variables as $V_i=(J^{i-1},Z_{i+1}^n,F_{1:r})$ and $U_i=K_A$. The following chain of equations follows:
\begin{align}
H(K_A)&=I(K_A;Z^n,F_{1:r})+H(K_A|Z^n,F_{1:r})\\
&\overset{(a)}{\leq} n\epsilon_{n}+H(K_A|Z^n,F_{1:r})\nonumber\\
&= n\epsilon_{n}+ H(K_A|Z^n,F_{1:r})-H(K_A|J^n,F_{1:r})+H(K_A|J^n,F_{1:r})\nonumber\\
&\overset{(b)}{=} n\epsilon_{n}+\sum_{i=1}^{n}I(K_A;J_i|J^{i-1},Z_{i+1}^{n},F_{1:r})-I(K_A;Z_i|J^{i-1},Z_{i+1}^{n},F_{1:r})+H(K_A|J^n,F_{1:r})\nonumber\\
&\overset{(c)}{=}n\epsilon_{n}+n\delta_n+\sum_{i=1}^{n}I(U_i;J_i|V_i)-I(U_i;Z_i|V_i)+I(K_A;K_B|J^n,F_{1:r})\nonumber\\
&\overset{(d)}{\leq} n\epsilon_{n}+n\delta_n+\sum_{i=1}^{n}I(U_i;J_i|V_i)-I(U_i;Z_i|V_i)+I(X^n;Y^n|J^n,F_{1:r})\nonumber\\
&\overset{(e)}{\leq} n\epsilon_{n}+n\delta_n+\sum_{i=1}^{n}I(U_i;J_i|V_i)-I(U_i;Z_i|V_i)+I(X^n;Y^n|J^n)\nonumber\\
&\leq n\epsilon_{n}+n\delta_n+\sum_{i=1}^{n}I(U_i;J_i|V_i)-I(U_i;Z_i|V_i)+nI(X;Y|J)\label{ineq:1}
\end{align}
where $\epsilon_n,\delta_n\rightarrow 0$ as $n\rightarrow\infty$; steps (a) and (c) are derived using Fano's inequality, (b) follows from the Marton–Körner identity, and (d) is a consequence of the Data Processing Inequality (DPI). To establish 
(e), we proceed as follows:
\begin{align}
    I(X^n;Y^n|J^n)&\overset{(a)}{=} I(X^n,F_1;Y^n|J^n)\overset{(b)}{\geq} I(X^n;Y^n|J^n,F_1)\overset{(c)}{=} I(X^n;Y^n,F_2|J^n,F_1)\overset{(d)}{\geq}
    I(X^n;Y^n|J^n,F_1,F_2)\nonumber\\
    &= I(X^n,F_3;Y^n|J^n,F_1,F_2)\geq I(X^n;Y^n|J^n,F_1,F_2,F_3)= \cdots\geq I(X^n;Y^n|J^n,F_{1:r}),\nonumber
\end{align}
where (a) and (b) follow from the Markov chains $(Y^n,F^{i-1})\mkv X^n\mkv F_i$ for odd $i$ and $(X^n,F^{i-1})\mkv Y^n\mkv F_i$ for even $i$, respectively. Also, (b) and (d) are derived from the chain rule.

Let $Q$ be uniformly distributed over $\{1,2,\cdots,n\}$ and independent of $X,Y,Z,J,F_{1:r}$. Define the random variables $V=(V_Q,Q)$, $U=U_Q$, $Z=Z_Q$, and $J=J_Q$. Under these definitions, inequality \eqref{ineq:1} simplifies to:  
\begin{align}
      \frac{H(K_A)}{n}\leq \epsilon_{n}+\delta_n+
      I(X;Y|J)+I(U;J|V)-I(U;Z|V).\label{eq:q1}
\end{align}
Additionally, the auxiliary random variables $V_i=(J^{i-1},Z_{i+1}^n,F_{1:r})$ and $U_i=K_A$ imply the Markov chain relationship $(U,V)\mkv (X,Y)\mkv (J,Z)$. 
Assuming that $n$ is sufficiently large, to prove constraint \eqref{cons1} for the random variables
$V_i,J_i,Z_i$, and $F_{1:r}$ as defined above, we proceed as follows:
\begin{align}
      I(V;Z)=\frac{1}{n}\sum_{i=1}^{n}I(V_i;Z_i)&=\frac{1}{n}\sum_{i=1}^{n}I(J^{i-1},F_{1:r};Z_i|Z_{i+1}^{n})
      \\&=\frac{1}{n}\sum_{i=1}^{n}I(F_{1:r};Z_i|Z_{i+1}^{n})+\frac{1}{n}\sum_{i=1}^{n}I(J^{i-1};Z_i|Z_{i+1}^{n},F_{1:r})\nonumber\\
      &\overset{(a)}{=}\frac{1}{n}\sum_{i=1}^{n}I(F_{1:r};Z_i|Z_{i+1}^{n})+\frac{1}{n}\sum_{i=1}^{n}I(Z_{i+1}^{n};J_i|J^{i-1},F_{1:r})\nonumber\\
      &\overset{(b)}{\leq} \frac{1}{n}I(Z^n;F_{1:r})+ \frac{1}{n}\sum_{i=1}^{n}I(J^{i-1},Z_{i+1}^{n},F_{1:r};J_i)\nonumber\\
      &=\frac{1}{n}I(Z^n;F_{1:r}) +I(V;J)
      =I(V;J)+H(Z)-\frac{1}{n}H(Z^n|F_{1:r})\nonumber\\
      &\overset{(c)}{\leq} I(V;J)+H(Z)-\Delta+\epsilon'_n\nonumber,
\end{align}
where step (a) follows from the Marton–Körner identity, 
(b) is derived using the chain rule, and 
(c) utilizes the assumption
$$\frac{1}{n}H(Z^n|F_{1:r})\geq \Delta-\epsilon'_n$$
for a vanishing sequence $\epsilon'_n$.
Thus, for any $J$ and a specific choice of auxiliary random variables 
$U$ and $V$ that satisfy the Markov chain 
$(U,V)\mkv (X,Y)\mkv (J,Z)$
and inequality \eqref{cons1}, by taking the limit as $n\rightarrow\infty$ in equation \eqref{eq:q1}, we obtain:
\begin{align*}
      S_{\Delta}(X;Y\|Z)\leq 
      I(X;Y|J)+I(U;J|V)-I(U;Z|V).
\end{align*}

The proof of the cardinality bounds is standard, and we omit it.
The second part of this theorem, for the special case of $Z=f(X,Y)$, follows from Lemma \ref{helper} (after applying the appropriate change of variables).
\end{proof}

\begin{lemma}\label{helper}
Suppose that $Z=g(X)$. For every joint distribution $P_{V,X,Y}$, the maximizer of $I(U;Y|V)-I(U;Z|V)$ over all $P_{U|V,X}$ is achieved when $H(X|U,V,Z)=0$.
\end{lemma}
\begin{proof}[Proof of Lemma \ref{helper}] 
Take some arbitrary tuple $P_{U|V,X}$. By the functional representation lemma, we can find a conditional distribution $P_{W|U,V,Z,X}$ such that $I(Z,U,V;W)=0$ and $H(X|Z,U,V,W)=0$. Now, consider the joint distribution $P_{W|U,V,Z,X}P_{Y|X}$.
Let $U'=(U,W)$. We claim that changing $U$ to $(U,W)$ would not decrease the expression, and moreover $H(X|Z,U',V)=0$. This would complete the proof. We need to show that
\begin{align}
    I(U;Y|V)-I(U;Z|V)\leq I(U';Y|V)-I(U';Z|V)
\end{align}
Equivalently,
\begin{align}
    0\leq I(W;Y|U,V)-I(W;Z|U,V)\label{conn}
\end{align}
which holds since $I(W;U,V,Z)=0$.   
\end{proof}

\section{An equivalent form of Conjecture \ref{conjecture:1}}
\label{appendixNew}
{\color{black}
Note that by setting $Z=Z'$, $Y=Y'$, and replacing $X$ by $(X,X')$ in Conjecture \ref{conjecture:1}, we obtain the following equation (as a special case):
\begin{align}
    S(X,X';Y'\|Z')-S(X';Y'\|Z')
    \leq I(Y',Z';X|X'), \qquad \forall P_{X,X',Y',Z'}.\label{eqnSPCconj}
\end{align}
Observe that $S(X,X';Y'\|Z')\geq S(X,X';X,Y'\|X,Z')$ because the first terminal can publicly reveal $X$ to all parties. Therefore, \eqref{eqnSPCconj} can be weakened to
\begin{align}
    S(X,X';X,Y'\|X,Z')-S(X';Y'\|Z')
    \leq I(Y',Z';X|X'), \qquad \forall P_{X,X',Y',Z'}.\label{eqnSPCconj22}
\end{align}
Therefore, if Conjecture \ref{conjecture:1} is true, \eqref{eqnSPCconj22} holds. We now claim the reverse.
\begin{theorem}
    Conjecture \ref{conjecture:1} holds if and only if \eqref{eqnSPCconj22} holds.
\end{theorem}
\begin{proof}
    
We first show that the conjecture holds if \eqref{eqnSPCconj} holds. The reason is as follows: since $S(X;Y\|Z)=S(Y;X\|Z)$, from \eqref{eqnSPCconj}, we also obtain the symmetric form:
\begin{align}
    S(X';Y,Y'\|Z')-S(X';Y'\|Z')
    \leq I(X',Z';Y|Y').\label{eqnSPCconj-form2}
\end{align}
In particular, \eqref{eqnSPCconj-form2} implies
\begin{align}
    S(X,X';Y,Y'|Z')-S(X,X';Y'|Z')
    \leq I(X,X',Z';Y|Y').\label{eqnSPCconj2}
\end{align}
Adding \eqref{eqnSPCconj} and \eqref{eqnSPCconj2}, we obtain
\begin{align}
    S(X,X';Y,Y'|Z')-S(X';Y'|Z')
    \leq I(Y',Z';X|X')+I(X',Z';Y|Y')+I(X;Y|X',Y',Z').\label{eqnSPCconj3}
\end{align}
Finally, note that 
\begin{align}
    S(X;Y\|Z)-S(X';Y'|Z')&=
    S(X;Y\|Z)-S(X,X';Y,Y'|Z')+S(X,X';Y,Y'|Z')-S(X';Y'|Z')\nonumber
    \\&
    \leq S(X;Y\|Z)-S(X;Y|Z')+S(X,X';Y,Y'|Z')-S(X';Y'|Z')\nonumber
    \\&
    \leq S_{\text{ow}}(X,Y\rightarrow Z'|Z)+S(X,X';Y,Y'|Z')-S(X';Y'|Z')\label{eqnexplnewf}
    \\&
    \leq S_{\text{ow}}(X,Y\rightarrow Z'|Z)+I(Y',Z';X|X')+I(X',Z';Y|Y')+I(X;Y|X',Y',Z')\label{eqnexplnewf2}
\end{align}
where \eqref{eqnexplnewf} follows from \eqref{eqnd1}, and \eqref{eqnexplnewf2} follows from \eqref{eqnSPCconj3}.

It remains to show that \eqref{eqnSPCconj22} implies \eqref{eqnSPCconj}.
Equation \eqref{eqnSPCconj} follows if one can prove the following: 
   $$S(X;Y\|Z)\leq \inf_{\tilde X\mkv X\mkv (Y,Z)}S(\tilde X;Y\|Z)+I(X;Y,Z|\tilde{X}).$$

Let 
\begin{align}
\Phi(X;Y\|Z)=\inf_{\tilde X\mkv X\mkv (Y,Z)}S(\tilde X;Y\|Z)+I(X;Y,Z|\tilde{X}).\label{defPhi1}
\end{align}
We use the approach in \cite{gohari-terminal1}. To prove $S(X;Y\|Z)\leq \Phi(X;Y\|Z)$, choose an arbitrary $(n,R,\delta)$ code, and let $W_1$ and $W_2$ be the private randomness used by the two parties, i.e., $H(F_i|X^n,F^{i-1},W_1)=0$ for odd $i$ and $H(F_i|Y^n,F^{i-1},W_2)=0$ for even $i$. Then, we have 
\begin{align}\label{eqn(a)}
\frac{1}{n}H(K_B)
&= \frac{1}{n}I(X^n,F_{1:r},W_1;K_B)+\epsilon_{n}
\\&\label{eqn(b)}= \frac{1}{n}(I(X^n,F_{1:r},W_1;K_B)-I(K_B;Z^n,F_{1:r}))+2\epsilon_{n}
\\&\label{eqn(c)}\leq  \frac{1}{n}\Phi(X^n,F_{1:r},W_1;K_B\|Z^n,F_{1:r})+2\epsilon_{n}\\
&\label{eqn(d)}\leq \frac{1}{n}\Phi(X^n,F_{1:r},W_1;Y^n,F_{1:r},W_2\|Z^n,F_{1:r})+2\epsilon_{n}
\\
&\label{eqn(e)}\leq \frac{1}{n}\Phi(X^n,F_{1:r-1},W_1;Y^n,F_{1:r-1},W_2\|Z^n,F_{1:r-1})+2\epsilon_{n}\\
&\label{eqn(f)}\leq \frac{1}{n}\Phi(X^n,F_{1:r-2},W_1;Y^n,F_{1:r-2},W_2\|Z^n,F_{1:r-2})+2\epsilon_{n}\\
&\leq \cdots
\\
&\label{eqn(g)}\leq \frac{1}{n}\Phi(X^n,W_1;Y^n,W_2\|Z^n)+2\epsilon_{n}
\\
&\label{eqn(h)}=\frac{1}{n}\Phi(X^n;Y^n\|Z^n)+2\epsilon_{n}\\
&\label{eqn(i)}\leq \Phi(X;Y\|Z)+2\epsilon_{n}
\end{align}
where $\epsilon_n$ converges to zero as $\delta$ converges to zero, \eqref{eqn(a)} follows from the fact that $K_A=K_B$ with probability $1-\delta$ and $H(K_A|X^n,F_{1:r},W_1)=0$, \eqref{eqn(b)} follows from the privacy constraint, \eqref{eqn(c)} from property iv in Lemma \ref{lemma14ext}, \eqref{eqn(d)} follows from property iii in Lemma \ref{lemma14ext} and the fact that $K_B$ is a function of $(Y^{n},W_2,F_{1:r})$, \eqref{eqn(e)}--\eqref{eqn(g)} follow from property ii in Lemma \ref{lemma14ext} and the fact that $H(F_i|X^n,F^{i-1},W_1)=0$ for odd $i$ and $H(F_i|Y^n,F^{i-1},W_2)=0$ for even $i$, \eqref{eqn(h)} follows from property v in Lemma \ref{lemma14ext}, and \eqref{eqn(i)} follows from property i in Lemma \ref{lemma14ext}.

\end{proof}

\begin{lemma}\label{lemma14ext}
    Assuming \eqref{eqnSPCconj22}, the function $\Phi$ as defined in \eqref{defPhi1} satisfies the following properties:
    \begin{enumerate}[i)]
        \item $\Phi(X^n;Y^n\|Z^n)\leq n \Phi(X;Y\|Z)$ when $P_{X^n,Y^n,Z^n}=\prod_{i=1}^{n}P_{X_i,Y_i,Z_i}$ is i.i.d.,
        \item $\Phi(F,X;F,Y\|F,Z)\leq \Phi(F,X;Y\|Z)$ and $\Phi(F,X;F,Y\|F,Z)\leq \Phi(X;F,Y\|Z)$ for any arbitrary $P_{F,X,Y,Z}$,
        \item $\Phi(X;Y'\|Z)\leq  \Phi(X;Y\|Z)$ when $H(Y'|Y)=0$,
        \item $\Phi(X;Y\|Z)\geq I(X;Y)-I(Y;Z)$,
        \item $\Phi(X,W_1;Y,W_2\|Z)=\Phi(X;Y\|Z)$ if $P_{W_1,W_2,X,Y,Z}=P_{W_1}P_{W_2}P_{X,Y,Z}$.
    \end{enumerate}
\end{lemma}
\begin{proof}
    Observe that $S(X;Y\|Z)$ satisfies similar properties as above, i.e., $S(X^n;Y^n\|Z^n)=nS(X;Y\|Z)$, and $S(F,X;F,Y\|F,Z)\leq S(F,X;Y\|Z)$, etc.; see \cite{gohari-terminal1}.

\emph{Proof of property i)} Let $\epsilon>0$ be arbitrary. Take some $\tilde{X}$ satisfying $\tilde X\mkv X\mkv (Y,Z)$ such that $$S(\tilde X;Y\|Z)+I(X;Y,Z|\tilde X)\leq \Phi(X;Y\|Z)+\epsilon.$$
Take $n$ i.i.d. copies of $(X,Y,Z,\tilde X)$ and denote them by $(X^n,Y^n,Z^n,\tilde X^n)$. We have:
\begin{align}
    \Phi(X^n;Y^n\|Z^n)\leq S(\tilde X^{n};Y^{n}\|Z^n)+I(X^{n};Y^{n},Z^{n}|\tilde X^{n})=n(S(\tilde X;Y\|Z)+I(X;Y,Z|\tilde X))\leq 
    n\Phi(X;Y\|Z)+n\epsilon.
\end{align}

\emph{Proof of property ii)} We first show that $\Phi(F,X;F,Y\|F,Z)\leq \Phi(F,X;Y\|Z)$. 
Let $\epsilon>0$ be arbitrary. Take some $\tilde{X}$ satisfying $\tilde X\mkv (F,X)\mkv (Y,Z)$ such that $$S(\tilde X;Y\|Z)+I(F,X;Y,Z|\tilde X)\leq \Phi(F,X;Y\|Z)+\epsilon.$$
The Markov chain $\tilde X\mkv (F,X)\mkv (Y,Z)$ implies the Markov chain $(F,\tilde X)\mkv (F,X)\mkv (F,Y,Z)$. Let $X'=(F,\tilde X)$. We have 
\begin{align}
   \Phi(F,X;F,Y\|F,Z)&\leq  S(X';F,Y\|F,Z)+I(F,X;F,Y,Z|X')
   \nonumber\\&=S(F,\tilde X;F,Y\|F,Z)+I(F,X;F,Y,Z|F,\tilde X)
  \nonumber \\&=S(F,\tilde X;F,Y\|F,Z)+I(F,X;Y,Z|\tilde X)-I(F;Y,Z|\tilde{X})
  \nonumber \\&\leq 
   S(\tilde X;Y\|Z)+I(F,X;Y,Z|\tilde X)\label{stepreq3}
   \\&\leq \nonumber\Phi(F,X;Y\|Z)+\epsilon.
\end{align}
where \eqref{stepreq3} follows from the assumption in \eqref{eqnSPCconj22}.

Next, we show that $\Phi(F,X;F,Y\|F,Z)\leq \Phi(X;F,Y\|Z)$. 
Let $\epsilon>0$ be arbitrary. Take some $\tilde{X}$ satisfying $\tilde X\mkv X\mkv (F,Y,Z)$ such that $$S(\tilde X;F,Y\|Z)+I(X;F,Y,Z|\tilde X)\leq \Phi(X;F,Y\|Z)+\epsilon.$$
From $\tilde X\mkv X\mkv (F,Y,Z)$, we deduce $(F,\tilde X)\mkv (F,X)\mkv (F,Y,Z)$ holds. Let $X'=(F,\tilde X)$. Thus, we have:
\begin{align*}
   \Phi(F,X;F,Y\|F,Z)&\leq  S(X';F,Y\|F,Z)+I(F,X;F,Y,Z|X')\\&=S(F,\tilde X;F,Y\|F,Z)+I(F,X;F,Y,Z|F,\tilde X)
   \\&=S(F,\tilde X;F,Y\|F,Z)+I(X;F,Y,Z|\tilde X)-I(X;F|\tilde{X})
   \\&\leq S(\tilde X;F,Y\|Z)+I(X;F,Y,Z|\tilde X)\\&\leq 
   \Phi(X;F,Y\|Z)+\epsilon.
\end{align*}
\emph{Proof of Property iii)}
Let $\epsilon>0$ be arbitrary. Take some $\tilde{X}$ satisfying $\tilde X\mkv X\mkv (Y,Z)$ such that $$S(\tilde X;Y\|Z)+I(X;Y,Z|\tilde X)\leq \Phi(X;Y\|Z)+\epsilon.$$
We have
\begin{align}
   \Phi(X;Y'\|Z)&\leq S(\tilde X;Y'|Z)+I(X;Y',Z|\tilde X)\leq S(\tilde X;Y|Z)+I(X;Y,Z|\tilde X) \leq \Phi(X;Y\|Z)+\epsilon.
\end{align}
\emph{Proof of Property iv)} We know that $S(X';Y\|Z)\geq I(X';Y)-I(Y;Z)$. Thus, $S(X';Y\|Z)+I(X;Y,Z|X')\geq I(X';Y)-I(Y;Z)+I(X;Y,Z|X')=I(X;Y)-I(Y;Z)+I(X;Z|X',Y)\geq I(X;Y)-I(Y;Z).$ This completes the proof.

\emph{Proof of Property v)} Take some arbitrary $P_{X'|X}$ reaching within $\epsilon$ of the infimum defining $\Phi(X;Y|Z)$. If we define $(W_1,W_2)$ independent of $(X',X,Y,Z)$, the Markov chain $X'\mkv(W_1,X)\mkv (W_2,Y,Z)$ will hold. 
Thus, $$\Phi(X,W_1;Y,W_2\|Z)\leq S(X';W_2,Y\|Z)+I(W_1,X;W_2,Y,Z|X')=S(X';Y\|Z)+I(X;Y,Z|X')\leq \Phi(X;Y|Z)+\epsilon.$$ This implies $\Phi(X,W_1;Y,W_2\|Z)\leq \Phi(X;Y\|Z)$. Next, to show $\Phi(X;Y\|Z)\leq \Phi(X,W_1;Y,W_2\|Z)$, take some $X'$ satisfying $X'\mkv(W_1,X)\mkv (W_2,Y,Z)$. Then, we have
$X'\mkv X\mkv (Y,Z)$
and
$$S(X';W_2,Y\|Z)+I(W_1,X;W_2,Y,Z|X')\geq S(X';Y\|Z)+I(X;Y,Z|X')\geq \Phi(X;Y\|Z).$$
\end{proof}
}

\bibliographystyle{IEEEtran}
\bibliography{mybiblio}

\end{document}